\documentclass[a4paper,11pt]{article}
\usepackage[utf8]{inputenc}
\makeatother

\usepackage{lmodern}
\usepackage{mathpazo}
\usepackage[dvipsnames]{xcolor}
\usepackage{fullpage}

\usepackage{amsmath}
\usepackage{amssymb}
\usepackage{amsthm}
\usepackage{mathrsfs}
\usepackage{amsfonts}
\usepackage{cases}

\newtheorem{definition}{Definition}
\newtheorem{proposition}{Proposition}
\newtheorem{theorem}{Theorem}

\newtheorem{lemma}{Lemma}

\newtheorem{assumption}{Assumption}

\theoremstyle{remark}
\newtheorem{example}{Example}
\theoremstyle{remark}
\newtheorem{remark}{Remark}

\usepackage{listings}

\usepackage[capposition=top]{floatrow}
\usepackage{algorithm}
\usepackage{algorithmic}
\usepackage{import}
\usepackage{comment}
\usepackage{natbib}
\usepackage{ascmac} % box
\usepackage{mathtools} % coloneqq
\usepackage{bbm}            % indicator function
\usepackage{fancybox}
\usepackage[at]{easylist}
\usepackage{multicol}
\usepackage{appendix}
\usepackage{hyperref}
\hypersetup{
    colorlinks=true,
    linkcolor=blue,
    filecolor=magenta,      
    urlcolor=cyan,
}
\usepackage{enumitem}
\usepackage{amsmath}
\usepackage[normalem]{ulem}
\usepackage{caption}
\usepackage{subcaption}

\usepackage{empheq} 

\newcommand{\mA}{\mathcal{A}}

\newcommand{\mE}{\mathcal{E}}

\newcommand{\mM}{\mathcal{M}}

\newcommand{\E}{\mathbbm{E}}

\newcommand{\argmax}{\mathop{\rm arg~max}\limits}
\newcommand{\argmin}{\mathop{\rm arg~min}\limits}

\DeclareMathOperator{\supp}{\mathrm{supp}}

\renewcommand{\tilde}{\widetilde}

\renewcommand{\phi}{\varphi}
\renewcommand{\epsilon}{\varepsilon}

\definecolor{C_pink}{RGB}{255, 64, 159}
\definecolor{C_Orange}{RGB}{204,85,0}
\definecolor{magenta}{RGB}{95,2,31}

\usepackage{graphicx}
\usepackage{amsmath}
\usepackage{amssymb}
\usepackage{mathtools}
\usepackage[dvipsnames]{xcolor}
\usepackage{calc}
\usepackage{caption, subcaption, floatrow}

\usepackage{amsopn}

\usepackage{subfiles}

\usepackage{tikz}

\usepackage{setspace}

\usepackage[a4paper, lmargin=0.1\paperwidth, rmargin=0.1\paperwidth, tmargin=0.1111\paperheight, bmargin=0.1111\paperheight]{geometry}
\title{Adversarial Elicitation}

\author{Andrei Iakovlev}

\begin{document}

%\begin{titlepage}
%	\centering % Center all text on the page
	
	%------------------------------------------------
	% Title
	%------------------------------------------------
	
%	\vspace*{2cm} % Add vertical space at the top
	
%	{\Huge\bfseries Adversarial Elicitation\par} % Paper title
	
%	\vspace{2.5cm} % Vertical space between title and author
	
	%------------------------------------------------
	% Author
	%------------------------------------------------
	
%	{\Large Andrei Iakovlev\par} % Author's name
	
%	\vspace{0.5cm} % Vertical space between name and institution
	
%	{\large\textit{Northwestern Unversity}\par} % Author's institution
	
%	\vfill % Pushes the following content to the bottom
	
	%------------------------------------------------
	% Advisors
	%------------------------------------------------
	
%	\begin{minipage}{0.6\textwidth} % Creates a box for the advisors
%		\centering
%		\large
%		\textbf{Advisors:} \\
%		Professor Jeffrey Ely
%		(Northwestern University) (Co-chair) \\
%		Professor Annie Liang
%		(Northwestern University) (Co-chair)\\ 
%		Professor Benjamin Golub (Northwestern University)
%	\end{minipage}
	
%	\vspace{1.5cm} % Vertical space between advisors and date
	
	%------------------------------------------------
	% Date
	%------------------------------------------------
	
%	{\large \today\par} % Use \today for the current date, or write a specific date
	
%	\vspace{1cm} % Add some space at the bottom
	
%	{\small Job Market Paper \\ \href{mailto:andrei.iakovlev@u.northwestern.edu}{andrei.iakovlev@u.northwestern.edu} \\ \href{https://sites.northwestern.edu/andreiiakovlev}{sites.northwestern.edu/andreiiakovlev}}
	
%\end{titlepage}

\maketitle

\begin{abstract}
	%\cmt{The first sentence is not as exciting, plus babbling comes to mind and trivializes the question. }
	When multiple informative equilibria are possible in a general cheap talk game, how much information can a principal \emph{guarantee} herself? To answer this question, I define the notion of worst-case implementation—implementation via the worst non-trivial equilibrium of a mechanism. Under this objective, standard full-commitment mechanisms fail, yielding the principal no more than her no-communication payoff. Partial commitment, however, can provide a strict improvement. The possibility of facing a strategic, uncommitted principal disciplines the agent's reporting incentives across all equilibria. I characterize the worst-case optimal mechanism and payoff under weak assumptions on the players' preferences. The optimal mechanism has a simple two-message structure. The agent's messages are \emph{polarizing}, designed to maximize their strategic impact on the uncommitted principal's actions. If full commitment is interpreted as decision automation, these results highlight a fundamental complementarity between automated and human decision-makers: the presence of a human aligns the agent's incentives to reveal information, while the automated system leverages these informative reports to take accurate actions. This strategic interaction is often overlooked by literature that compares the two based on standalone decision accuracy. Applications of the model include bail-setting automation, fintech lending, delegation, lobbying, and audit design.
		\end{abstract}

\section{Introduction}

Should decisions be guided by rigid rules or by human discretion? This classic question is central to institutional design, from judicial sentencing to automated lending. A commitment to a fixed rule, or mechanism, offers a powerful advantage: it allows a principal (the decision maker) to manage the incentives of an agent (the informed party) and steer the outcome to her most preferred one. This insight is the foundation of modern mechanism design. Yet, commitment is also a source of fragility. Performance of a committed rule relies on the agent following the principal-preferred strategy. This means that a rigid rule can be exploited if the agent follows a different, payoff-reducing strategy. Such a rule-bound principal is vulnerable to being \emph{gamed}. This raises a crucial question: should a principal who is concerned about such adversarial play fully commit to a decision rule, or does retaining discretion provide a valuable safeguard?

To investigate this question I study a principal-agent model with cheap talk. The agent  has a state-independent payoff (similar to \cite{lipnowski2020cheap} and \cite{partial_lipnowski2022}). The principal can commit to an action plan and choose the probability with which her commitment to the plan binds. With the remaining probability, she retains discretion to act strategically. I evaluate each pair (action plan, commitment probability) by its guarantee --- the lowest equilibrium payoff it can deliver to the principal. My question is: what is the largest guarantee and the mechanism that attains it?

What makes this question challenging is that the standard techniques of mechanism design offer little help in characterizing the guarantee. Unlike the question of optimality, which can be fully resolved by focusing on full commitment (see Section \ref{sec:benchmark}), finding the largest guarantee requires extending the search beyond full commitment, since full commitment mechanisms offer surprisingly low securable payoffs. This is because when the principal fully commits, an agent can always choose to make his messages uninformative ('babble'). This leaves the principal unable to guarantee any improvement over having no information at all. 

My central insight is that the principal can guarantee herself a better outcome by committing partially. I characterize the mechanism and level of commitment with the largest guarantee, which I call worst-case optimal. I show that under weak assumptions on the player's preferences, plus a regularity condition, the worst-case optimal mechanism has three key properties:

\begin{easylist}[itemize]
	@  the communication is coarse: the principal optimally elicits a yes-no report from the agent, regardless of the richness of the agent's private information. This is in contrast to standard optimality, which generally requires the agent to report as many messages as there are states\footnote{More precisely, standard optimality requires the agent to report as many messages as there are actions to be implemented.}.  
	@  the agent's reporting strategy is polarizing: the agent's messaging maximizes the variance in his own utility, should the principal best-respond to his messages. Less informative strategies, which plague mechanisms with full commitment, are not incentive-compatible under the worst-case optimal mechanism. 
	@ optimal commitment is necessarily partial but extensive. The rare discretionary interventions serve to shape the agent's incentives to be informative. The majority of decisions are delegated to the committed rule, which can rely on informative communication from the agent even in the worst equilibrium. 
	\end{easylist}

The model admits a primary interpretation in the context of algorithmic decision-making. In this interpretation, the principal's commitment corresponds to the use of an algorithmic decision rule, while her retention of discretion corresponds to human oversight or intervention. The commitment probability, $\chi$, represents the fraction of cases adjudicated by the algorithm, which follows the pre-specified plan $\pi$. Consequently, full commitment ($\chi=1$) corresponds to a fully automated system, while no commitment ($\chi=0$) corresponds to a system relying exclusively on human decision-makers.

Under this interpretation, the structure of the worst-case optimal mechanism reveals a strategic complementarity between human and algorithmic decision-makers. The two play distinct and mutually reinforcing roles in eliciting information. The possibility of facing a strategic human, who will form posterior beliefs and best-respond, disciplines the agent's reporting incentives across all equilibria. This threat of strategic interpretation compels the agent to adopt an informative messaging strategy. The algorithmic component, in turn, leverages these informative reports to execute the committed action plan with precision, thereby realizing the gains from the elicited information. This analysis suggests that a singular focus on the standalone decision accuracy of humans versus algorithms, common in the existing literature (\cite{bail_kleinberg2018}, \cite{insruance_fraud_aslam2022}, \cite{insurance_aihumanmix_gao2023}), overlooks the critical role that human discretion plays in shaping the information environment in which both decision-makers operate.

This paper is organized as follows. Section 2 formally introduces the principal-agent model of cheap talk. The model is characterized by two key assumptions: the agent's payoff is state-independent, and it is strictly monotone in the principal's continuous action. This section also defines a mechanism with partial commitment and the solution concept, perfect Bayesian equilibrium, restricted to non-trivial communication involving at least two on-path messages. As a benchmark, Proposition \ref{prop:full_commitment} characterizes the optimal mechanism, which involves full commitment. This characterization is expressed in terms of the players' \emph{mutual payoff} function, which specifies the principal's payoff for a given level of utility delivered to the agent in each state. The optimal payoff, $\overline V$, is then determined geometrically via the concavification of this function.

Section 3 proceeds to the central problem of worst-case implementation. A key result, Theorem \ref{thm:carrot_stick}, establishes that under generic assumptions, the worst-case optimal mechanism requires only a binary message structure. This implies that the optimal communication protocol reduces to the principal asking the agent a \emph{yes-no} question. In addition, Theorem \ref{thm:carrot_stick} shows that the optimal mechanism must be \emph{polarizing}. Such a mechanism creates maximal separation in the agent's payoffs from the committed plan, a structure sometimes referred to as a \emph{carrot-and-stick} approach. It is then shown that polarizing mechanisms restrict the agent's equilibrium behavior to a specific class of \emph{polarizing strategies}.

When the polarizing strategy is unique, a full characterization of the worst-case optimal mechanism and payoff is derived. The solution employs a geometric approach analogous to that of the full-commitment benchmark, but applied to polarized utility functions derived from the agent's messaging strategy. Finally, the analysis establishes a condition under which the worst-case optimal payoff, $V^\star$, coincides with the optimal payoff, $\overline{V}$. When the agent's private information is binary, the unique polarizing strategy is truth-telling. In this case, the principal can fully elicit the agent's information even in the worst-case equilibrium and thereby achieve the first-best outcome.

\section{Literature}
This paper contributes to several strands of literature, primarily concerning strategic information transmission, mechanism design, and decision automation.

The analysis is situated within the extensive literature on communication with soft information, which originates with the study of cheap talk (\cite{crawford1982strategic}). This literature has largely focused on two benchmarks: Bayesian persuasion, which characterizes the sender's optimal payoff under full commitment power (\cite{kamenica2011bayesian}), and cheap talk, which analyzes communication when the sender has no commitment ability (\cite{lipnowski2020cheap}). Shifting attention to the receiver, this paper examines the intermediate space where the principal (the receiver) has endogenous control over her degree of commitment. The model also relates to work studying the role of soft information in applied settings, such as credit and insurance markets (\cite{soft_agarwal2011role}, \cite{insurance_soft_li2021}). This work focuses on finite private information and, thus, does not consider the effects of multi-dimensional communication (\cite{cheap_talk_persuasion_chakraborty2010}, \cite{multi_cheap_multiple_battaglini2002}, \cite{multi_cheap_checking_ball2025}).

The paper contributes to the theory of mechanism design. The canonical approach in this field assumes the designer's full commitment to a direct mechanism, a paradigm applied in static (\cite{mechanism_design_myerson1981optimal}), dynamic (\cite{mechanism_design_dynamic0_bergemann2019dynamic}, \cite{mechanism_design_dynamic_pavan2014dynamic}), and information design settings (\cite{kamenica2011bayesian}). While some work considers partial commitment (\cite{partial_and_revelation_bester2001contracting}), this paper introduces the objective of optimal worst-case implementation. A central finding is that, under this objective, the designer benefits from partial commitment, and the optimal mechanism requires only a two-message report, regardless of the size of the agent's type space. This contrasts sharply with the standard revelation principle (\cite{mechanism_design_myerson1981optimal}, \cite{partial_and_revelation_bester2001contracting}). The analysis extends certain methodological tools, such as the geometric approach (\cite{kamenica2011bayesian}), to the partial commitment setting and develops alternatives for others, like the revelation principle, that do not directly apply.

The concept of worst-case implementation is new to the implementation literature. It is related to, but distinct from, the "price of anarchy," which compares the best and worst Nash equilibria of a given game (\cite{worst_equilibria_koutsoupias1999}). It is a weaker solution concept than full implementation, which requires all equilibria of a mechanism to yield a specific outcome (\cite{implementation_maskin_maskin1999nash}). The focus on robustness connects this paper to the literature on robust mechanism design, which considers implementation that is robust to uncertainty over the agent's private information (\cite{bergemann2016informationally}) or preferences (\cite{robust_bergemann2023}). In contrast, this paper focuses on robustness to equilibrium selection for a fixed economic environment. The analysis also complements recent work on persuasion with partial sender commitment (\cite{partial_lipnowski2022}, \cite{partial_bayesian_min2021}) by analyzing the problem from the perspective of the principal who controls her own commitment.

Finally, the model offers insights for several applied fields. It is relevant to the classic settings of delegation (\cite{delegation_dessein2002authority}, \cite{delegaion_alonso2008optimal}) and lobbying (\cite{lobbying_rulesofcommunication_austen1990information}). The leading application is to the literature on decision automation. This topic is particularly salient in fintech, with the rise of algorithmic products (\cite{fintech_buchak2018}) and attendant debates on discrimination (\cite{fintech_discrimination_bartlett2022}); in insurance, where AI is used for underwriting and fraud detection (\cite{insruance_fraud_aslam2022}, \cite{insurance_aiSurvey_bhattacharya2025}); and in medicine, with the development of diagnostic algorithms (\cite{medicine_rajpurkar2017chexnet}, \cite{development_medicine_gulshan2016}). Much of this literature compares the flat accuracy of algorithmic and human decision-makers (\cite{bail_kleinberg2018},\cite{underwriting_jansen2025}, \cite{insurance_aihumanmix_gao2023}, \cite{automation_iakovlev2025value}). My contribution is to provide a strategic rationale for mixed human-AI systems, showing how the presence of a discretionary human decision-maker can discipline an agent's incentives to reveal information to an automated system. This focus on the agent's incentives distinguishes the analysis from work that studies how AI oversight affects the principal's own cognitive process (\cite{aioversight_almog2024}).

\section{The Model}

\label{sec:model}

  There are two players: a principal (she) and an agent (he). The agent privately observes a state of the world $\theta\in \Theta$ and sends a costless message $m\in \mM$ to the principal. Both $\Theta$ and $\mM$ are finite, and $|\mM|\ge |\Theta|=n$. The state $\theta$ follows a distribution $\boldsymbol{\rho}\in\Delta(\Theta)$, which is commonly known. 

  After observing the agent's message, the Bayesian principal updates her beliefs about $\theta$ and takes an action $a\in\mA\equiv[0,1]$. The agent and the principal have continuous utility functions $u_A:\mA \rightarrow [0,1]$ and $u_P:\mA\times\Theta\rightarrow [0,1]$, where $u_A(a)$ is strictly monotone\footnote{The results extend when strict monotonicity is not imposed. The corresponding adjustment is discussed in footnote \label{ftn:monotone_cont}.} and normalized to $u_A(0)=0$ and $u_A(1)=1$. Both players aim to maximize their expected payoffs and, for convenience, actions $0$ and $1$ are the principal's optimal actions in some $\theta,\theta'\in\Theta$.

Prior to communication, the principal chooses a mechanism $Q=(\pi,\chi)$. The \emph{credibility level} $\chi\in[0,1]$ is the probability that the principal is bound by her committed action plan, $\pi=\{\pi_m\}_{m\in\mathcal{M}}$, where each $\pi_m\in\Delta(\mathcal{A})$. With probability $1-\chi$, the principal is uncommitted and selects an action to maximize her expected utility, conditional on her posterior beliefs $\mu(\cdot|m)$.

An equilibrium of the game induced by $Q$ consists of an agent's messaging strategy, $\sigma_A:\Theta\rightarrow\Delta(\mathcal M)$, and an uncommitted principal's action strategy, $\sigma_P:\mathcal M\rightarrow \Delta(\mathcal{A})$. The agent's expected payoff from sending message $m$ is given by:
\[U_A(m)=\chi\mathbb{E}_{\pi_m}[u_A(a)]+(1-\chi)\mathbb{E}_{\sigma_P(m)}[u_A(a)]\]
A strategy profile $( \sigma_A, \sigma_P)$ constitutes an equilibrium if it satisfies the following conditions:

\paragraph{Equilibrium.} An equilibrium of the game induced by $Q$ consists of an agent's messaging strategy, $\sigma_A:\Theta\rightarrow\Delta(\mathcal M)$, and an uncommitted principal's action strategy, $\sigma_P:\mathcal M\rightarrow \Delta(\mathcal{A})$. The agent's expected payoff from sending message $m$ is given by:
\[U_A(m)=\chi\mathbb{E}_{\pi_m}[u_A(a)]+(1-\chi)\mathbb{E}_{\sigma_P(m)}[u_A(a)]\]
A strategy profile $( \sigma_A, \sigma_P)$ constitutes an equilibrium if it satisfies the following conditions:

 %\footnote{A typical equilibrium messaging strategy is a mixing one. My model admits a purification argument, which lets me interpret such mixing strategies as the principal's ambiguity over how agents of different types resolve uncertainty via small payoff shocks.} 

\

\begin{enumerate}
\item \emph{Agent-Optimality.} For each state $\theta \in \Theta$, the agent's strategy $\sigma_A(\cdot|\theta)$ is supported on the set of messages that maximize his expected utility:
\[ \argmax_{\tilde m\in \mathcal M}U_A(\tilde m) \tag{$\text{AO}$} \label{eqm:ind}\]

\item \emph{Principal-Optimality.} For each on-path message $m$, the uncommitted principal's strategy $\sigma_P(\cdot|m)$ is supported on the set of actions that maximize her expected utility, given the posterior belief $\mu(\cdot|m)$ derived from $\sigma_A$ and the prior $\rho$:
\[\argmax_{\tilde a\in\mathcal{A}} \mathbb{E}_{\mu(\cdot|m)}[ u_P(\tilde a,\theta)] \tag{PO} \label{eqm:br}\]
For any off-path message, $\sigma_P$ must be a best response to some belief.

\item  \emph{No perfect pooling.} The ex-ante probability of each message is less than one: $\sum_\theta\sigma_A(m|\theta)\rho(\theta)<1$ for all $m\in\mathcal M$.
\end{enumerate}

 %In addition to standard equilibrium conditions, I require that $P(m)>0$ for all $m$. In a binary message case this condition filters out babbling equilibria, but does not exclude non-informative messaging. Indeed, since exchanging a single on-path message is equivalent to not having communication at all, by focusing on equilibria with $P(m)>0$ I restrict attention to scenarios with non-trivial message exchanges, whether or not the messages are informative\footnote{The argument for $|\mM|>2$ will go differently.}. In practice, there is strong evidence that agents tend to over-communicate in Cheap Talk games (Cai, Wang, 2006) including sending informative messages when only a babbling equilibrium should be possible. 

\

The agent-optimality and principal-optimality conditions are standard requirements for a Perfect Bayesian Equilibrium. The No Perfect Pooling condition is imposed for exposition. The presence of pooling equilibria in particular, and babbling equilibria in general, is a standard occurrence in cheap talk games. The family of mechanisms, studied in this paper, cannot eliminate the pooling equilibrium. Instead, my results show that the principal can achieve a discrete improvement upon the pooling equilibrium.  \footnote{There exists a perturbed version of the present model, in which the No Perfect Pooling condition can be replaced by focusing on neologism-proof equilibria (\cite{neologism_farrell1993meaning}).} This restriction is consistent with empirical evidence that suggests aversion to pooling strategies (\cite{nobabbling_empirical_cai2006}). The main results are robust to alternative refinements, such as restricting attention to equilibria with strictly positive informational context.\footnote{That is, if we exclude babbling equilibria.}

%The modification is adding value $\varepsilon$ to the agent's payoff from any (!) message as long as there are two or more messages. This does not affect any equilibria or any on-path payoffs (a+\varepsilon vs b+\varepsilon is equivalent to a vs b), but with this modification for any positive \varepsilon the pooling equilibrium is not neologism-proof (and possibly fails other similar refinements) in the polarizing mechanism sequence. This is because in the original model the agent's payoffs in the pooling and the polarizing equilibrium both converge to the same payoff $\E_\pi[u]$ as $\chi\rightarrow 1$. In the modified model, the agent's payoff from any on-path message in the polarizing equilibrium converge to $\E_\pi[u]+\varepsilon$, while the pooling payoff converges to $\E_\pi[u]$. This means that for $\chi$ sufficiently close to zero, the agent would strictly prefer to coordinate on the polarizing equilibrium. And since in this equilibrium the agent's are indifferent between the messages, the agent that keeps using the message with the largest committed payoff will not have incentives to mimic the deviating one.

To simplify the exposition, the principal's off-path beliefs are not explicitly modeled. The analysis proceeds under the assumption of pessimistic beliefs: any off-path message $m$ induces the uncommitted principal to take an action that yields the agent his lowest possible payoff. Under this assumption, the agent's expected payoff from an off-path deviation is determined entirely by the committed plan: $U_A(m) = \chi\mathbb{E}_{\pi_m}[u_A(a)]$. This specification does not affect the qualitative nature of the results.

  \paragraph{Implementation.} %Denote by $V_B$ is the principal's \emph{baseline payoff}, achieved when the players don't communicate:
	% \[V_B:=\max_{\Tilde a\in\mA}\E_{\boldsymbol\rho}[u_P(\Tilde a,\theta)]\]	 

	 Denote by $\mE(Q)$ the set of equilibria induced by $Q\equiv( \pi,\chi)$.  Say that a messaging strategy $\sigma_A$ is feasible under a mechanism $Q$ if it is a part of some equilibrium, induced by $Q$, so   $( \sigma_A,\sigma_P)\in\mE(Q)$ for some $\sigma_P$.  Denote the set of messaging strategies, feasible under $Q$ by $\Sigma(Q)$. I take the convention that if $\mathcal E(Q)=\varnothing$, then the principal's payoff from using $Q$ is her no communication payoff.
	 
	 For any mechanism $Q$ and a strategy profile $E=(\sigma_A,\sigma_P)$ denote the principal's expected payoff by
%\E_{Q,P}[u_P(a,\theta)]
\[V(Q,E)=\chi \sum_{\theta,m}\rho(\theta)\sigma_A(m|\theta)\E_{\pi_m}[u_P(a,\theta)]+(1-\chi)\sum_{\theta,m}\rho(\theta)\sigma_A(m|\theta)\mathbb E_{\sigma_P(m)}[u_P(a,\theta)]\]

\paragraph{Worst-case implementation. } This paper focuses on mechanisms that are robust to equilibrium selection. To formalize this I define the concept of worst-case implementation.

\begin{definition}[Worst-case implementation]
	A payoff $v$ is \emph{worst-case $\chi$-implementable} for $\chi<1$ if it constitutes the principal's lowest equilibrium payoff under some mechanism $Q=(\pi,\chi)$:
\[v=\inf_{E\in\mathcal E(Q)}V(Q,E)\]	
A payoff $v$ is \emph{worst-case $1$-implementable} if it is the limit of worst-case $\chi_n$-implementable payoffs $v_n$ for a sequence of commitment levels $\chi_n\rightarrow1$:
\[v=\lim_{n\rightarrow\infty} v_n\]
\end{definition}

The definition of worst-case $1$-implementability ensures that the set of worst-case implementable payoffs is closed.\footnote{
Ensuring closedness is important, since  the correspondence that maps a mechanism $Q$ into the set of the principal's equilibrium payoffs suffers a discontinuity at $\chi=1$  (Lemma \ref{lem:chi_around_1}). } Considering the limiting payoffs as implementable reflects the concept of virtual implementation, where an outcome can be arbitrarily approximated by the implementable ones of nearby mechanisms (\cite{virtual_abreu1992}, \cite{bergemann2016informationally}). The central object of the present analysis is the \emph{worst-case optimal payoff}, defined as the supremum over all worst-case $\chi$-implementable payoffs%\footnote{Without the No Perfect Pooling equilibrium condition any mechanism has at least one cheap talk equilibrium. Namely, when $\chi>0$, sending a message that grants the largest committed payoff, independently of the type, is an equilibrium. The subsequent analysis qualitatively extends to this case under two adjustments. First, the mechanisms should be compared lexicographically as follows. If two mechanisms have the same worst case payoff, then exclude the respective equilibrium from each mechanism and compare the two mechanisms based on the new infimum. Second, instead of maximizing the value function in Theorem \ref{thm:characterization_general}, the worst-case optimal payoff is computed at the no-communication optimum $u_\rho$. \label{ftn:lexicographic}} 
for $\chi \in [0,1]$:
\[V^\star=\sup_Q\inf_{E\in\mathcal E(Q)}V(Q,E)\]

In the next section I characterize the optimal payoff $\overline V$ using the standard mechanism design approach. The optimal mechanism and the optimal payoff admit a simple geometric characterization. The geometric approach, developed in Section \ref{sec:benchmark} will be helpful in analyzing worst-case implementation in Section \ref{sec:worst-case_implementaation}.

%\cmt{Refinement.} There are two ways to deal with the pooling equilibiurm. 

\label{sec:exp}

\section{Illustrative example}

%\begin{example}

\label{exp:socrates} 

This section demonstrates that by properly choosing an action plan a partially committing principal can guarantee that the agent reveals his information in any feasible equilibrium, which is a powerful tool for worst-case implementation. The state and message spaces are binary, $\Theta=\{\theta_l,\theta_h\}$ and $\mathcal{M}=\{l,h\}$, with prior $P(\theta=\theta_h)=\rho$. The agent's utility is linear in the action, $u_A(a)\equiv a$. The uncommitted principal's best response to a belief $\mu=P(\theta=\theta_h)$ is a strictly increasing function $a(\mu)$, with $a(0)=0$ and $a(1)=1$.\footnote{A strictly monotone $a(\mu)$ can be generated by strictly single-peaked principal's utilities, such as in the Crawford-Sobel model. The Socrates Effect extends beyond beyond linear $u_A$ and monotone $a(\mu)$. A simple sufficient condition for the Socrates Effect is that the principal takes actions $0$ or $1$ only when she knows the state is $\theta_l$ or $\theta_h$.}

Consider a \emph{naive mechanism} $Q_N$ defined by a commitment probability $\chi=1/2$ and committed actions $\pi(l)=0$ and $\pi(h)=1$, where, in abuse of notation, $\pi(m)$ represent deterministic actions. This mechanism induces a unique, fully revealing equilibrium.

\begin{lemma}[Socrates Effect]
The naive mechanism $Q_N$ induces a unique non-trivial equilibrium, in which the agent's strategy is separating:
\[ \sigma_A(m=h|\theta_l)=1\]
\[\sigma_A(m=l|\theta_h)=1\] 
\label{lem:socrates}
\end{lemma}

\textbf{Proof.}
In any non-trivial equilibrium, the agent must be indifferent between sending message $l$ and $h$. This indifference condition, $U_A(l)=U_A(h)$, expands to:
\[\frac{1}{2} \pi(l)+\frac{1}{2}a(\mu(l))=\frac{1}{2}\pi(h)+\frac{1}{2}a(\mu(h))\]
Substituting the committed actions $\pi(l)=0$ and $\pi(h)=1$ yields $a(\mu(l))-a(\mu(h))=1$. Since the range of $a(\mu)$ is $[0,1]$, this equality holds only if $a(\mu(l))=1$ and $a(\mu(h))=0$. Given that $a(\mu)$ is strictly increasing with $a(0)=0$ and $a(1)=1$, this requires the posterior beliefs to be exactly $\mu(l)=1$ and $\mu(h)=0$. These beliefs are consistent only with the specified separating strategy. \qed

The equilibrium strategy is counter-intuitive: the high type reports $l$ and the low type reports $h$. This messaging behavior is \emph{polarizing}—it induces the uncommitted principal to form extreme posterior beliefs and consequently take the most dispersed actions. Such polarizing strategies play a central role in the worst-case optimal mechanisms, as demonstrated in Section \ref{sec:worst-case_implementaation}.

This example highlights a central trade-off inherent to partial commitment: the balance between the ex-post optimality afforded by discretion and the incentive provision required for information revelation. The benefit of the mechanism materializes when the principal retains discretion; informed by the separating equilibrium, she selects her first-best action. The cost is borne when her commitment binds, as inducing separation requires committing to an action plan that is ex-post suboptimal—for instance, taking a suboptimal action $0$ in response to message $l$, which is known to come from the high type.

This equilibrium structure is sustained only at a precise level of commitment. If the principal's discretion was to increase (i.e., $\chi < 1/2$), the agent's indifference condition would be violated. The payoff from reporting $l$, $U_A(l)=1-\chi$, would strictly exceed the payoff from reporting $h$, $U_A(h)=\chi$, causing the separating equilibrium to unravel. Any resulting equilibrium must necessarily be less informative. To restore indifference, the agent's strategy must induce less extreme posterior beliefs, which in turn leads to less dispersed actions from the uncommitted principal, compensating for the increased weight on her strategic response. This reveals a fundamental constraint that the principal faces: an increase in her discretion must be met by a decrease in the informativeness of any equilibrium messaging strategy.

%By Lemma \ref{lem:cheaptalk}, any mechanism with full/abscent commitment admits a babbling equilibrium. 

%\end{example}

The preceding analysis, culminating in Lemma \ref{lem:cheaptalk} and the Socrates Effect of Example \ref{exp:socrates}, establishes a clear trade-off for the principal. Her choice of commitment level can be framed as a selection among three classes of mechanisms, each with a distinct profile of achievable payoffs and worst-case guarantees.

\begin{itemize}
	\item \emph{Full Commitment ($\chi=1$)}: This class offers the highest potential payoff, $\overline V$, but provides the weakest security. As established in Lemma \ref{lem:cheaptalk}, the existence of a babbling equilibrium means its guaranteed payoff cannot exceed the no-communication baseline, $V_B$.

\item \emph{No Commitment ($\chi=0$)}: This class restricts outcomes to those supportable in a cheap talk game. It offers neither the prospect of achieving a high payoff nor a guarantee better than $V_B$.

\item \emph{Partial Commitment} ($\chi \in (0,1)$): This class provides a solution to the limitations of the polar cases. While, as we will see later, it may sacrifice the ability to achieve the optimal payoff $\overline V$, its defining feature is the capacity to eliminate uninformative equilibria. By doing so, it can secure a worst-case payoff strictly greater than $V_B$. 
\end{itemize}

\section{Benchmarks}

%\label{sec:eqm}

%some aspects of equilibrium multiplicity. The requirement (SP) ensures that the principal has a unique best response to any posterior (see \cmt{Hart, Kremer, and Perry, 2017}). I also rule out saturation scenarios, such as the following example.

%\begin{example}
	%Suppose the agent has a binary type $\theta\in\{\theta_1,\theta_2\}$ and the principal takes action $a=1$ if her posterior belief about state $\theta_1$ exceeds $0.5$ and she takes action $a=0$ other\emph{wise} (\cite{kamenica2011bayesian}). Then (NS) does not hold. The principal's best response to any belief $\mu$ in a neighborhood of $\mu(\theta_1)=1$  with a radius of $0.5$ is $a(\mu)=1$, violating the first part of (NS).
%\end{example} 
 \label{sec:benchmark}

This section establishes two benchmarks for the subsequent analysis. The principal's no-communication payoff, $V_B$, is the largest payoff she can achieve by selecting an action based on her prior information:
	 \[V_B:=\max_{\Tilde a\in\mA}\E_{\boldsymbol\rho}[u_P(\Tilde a,\theta)]\]	 
The principal's optimal payoff, $\overline V$, is the largest payoff achievable across all mechanisms and all equilibria they might induce:
\begin{align}
\overline{V}:=\sup_{Q}\sup_{E\in\mathcal E(Q)} V(Q,E) \notag
\end{align}

The characterization of $\overline V$, formalized in Proposition \ref{prop:full_commitment}, is simplified by the revelation principle. It is sufficient to restrict the analysis to direct, full-commitment mechanisms ($\mathcal{M}\equiv \Theta$, $\chi=1$) that induce a truth-telling equilibrium. Given the agent's state-independent utility, this requires that the mechanism render the agent indifferent among all possible messages. Formally, the committed plan must satisfy the agent-optimality condition:

\begin{equation}
	\mathbb E_{\pi_m}[u_A(a)]=\mathbb E_{\pi_{m'}}[u_A(a)] \quad \forall m,m'\in\mathcal M\equiv \Theta\label{eqm:IC1} \tag{AO$_1$}
\end{equation}
To solve for the principal's best outcome under this constraint, it will be convenient to introduce
 \[\gamma_\theta(u):=u_P(u_A^{-1}(u),\theta)\]
 
 This function defines the principal's utility in state $\theta$ from the unique action that delivers an expected payoff of $u$ to the agent.\footnote{Uniqueness of this action follows from the assumption that $u_A$ is strictly monotone. All the results in the paper can be extended beyond the case of strictly monotone $u_A $, in which case one defines $\gamma_{\theta}(u)$ as the upper envelope of the graph of the correspondence $ u_P(u_A^{-1}(u),\theta)$, where $u_A^{-1}$ denotes the set of actions that solve $u_A(a)=u$. Defined this way, the correspondence $\gamma_{\theta}(u)$ is upper hemicontinuous. Instead of applying the Berge's maximum theorem (which needs continuous $\gamma$), one can apply its generalization that requires UHC $\gamma$ to show that the best response correspondence is UHC.} The maximal payoff the principal can achieve in state $\theta$ for a given agent utility $u$, allowing for randomization over actions, is given by the concave envelope of this function, denoted $\gamma_\theta^{co}(u)$.

 %By making the agent indifferent, the mechanism removes any incentive to lie. To solve for the principal's best outcome under this constraint, I reframe her problem using the mutual payoff function, $\gamma_{\theta}(u):=u_P(u_A^{-1}(u),\theta)$. This function captures the principal's essential trade-off, defining the utility the principal receives in a given state $\theta$ from an action that delivers a payoff of $u$ to the agent\footnote{The characterization can be extended beyond the case of strictly monotone $u_A $, in which case we define $\gamma_{\theta}(u)=\max u_P(u_A^{-1}(u),\theta)$.}. The principal's optimal payoff is then found by taking the concave hull of this function, a procedure that identifies the best possible outcome achievable through randomization, as formalized in Proposition \ref{prop:full_commitment}.

 %The Agent-optimality condition (\ref{eqm:ind}) coincides with (\ref{eqm:IC1}) when $\chi=1$ and the agent uses all messages on-path. The latter is a necessary condition for truth-telling. Maximizing the principal's payoff over the mechanisms that satisfy (\ref{eqm:IC1}) results in Proposition \ref{prop:full_commitment}.  Define the mutual payoff function in state $\theta$ as $\gamma_{\theta}(u):=u_P(u_A^{-1}(u),\theta)$ - the utility that the principal gets from delivering a payoff $u$ to the agent\footnote{The characterization can be extended beyond the case of strictly monotone $u_A $, in which case we define $\gamma_{\theta}(u)=\max u_P(u_A^{-1}(u),\theta)$.}. And denote by $\overline f(x)$ the concave hull of a function $f(x)$.  

 \begin{proposition}
 	\[\overline{V}=\max_{u\in[0,1]}\sum_{\theta\in\Theta}\rho_\theta\gamma_\theta^{co}(u)\] 
 	\label{prop:full_commitment}
 \end{proposition} 
 
 Proposition \ref{prop:full_commitment} establishes that the optimal mechanism is found by selecting a single expected utility level, $u$, to offer the agent across all reports. For each truthfully reported state $\theta$, the principal implements an action plan (potentially a lottery) that achieves the maximal payoff $\gamma_\theta^{co}(u)$ subject to delivering this common expected utility $u$ to the agent. While the agent's expected utility is constant across states, the lotteries need not be.

The characterization of the optimal payoff, $\overline V$, draws a close parallel to the classic Bayesian persuasion (BP) framework. Here, the agent's realized payoffs, $u_A(a)$, act as the receiver's posteriors in BP, while the action plans, $\pi_\theta$, correspond to distributions over these posteriors. The incentive constraint (\ref{eqm:IC1}) mirrors the familiar Bayes-plausibility condition. However, two key distinctions set this problem apart. First, instead of choosing a single distribution of posteriors consistent with a prior, the principal must select a set of action distributions $\{\pi_\theta\}_{\theta}$ that are mutually consistent (i.e. deliver the agent the same expected payoff $u$). Second, and more crucially, the agent's expected payoff $u$,  that parallels the prior in BP, is not a fixed parameter but the principal's choice variable. This is why the optimal payoff $\overline V$ results from an optimization over the incentives, $u$, offered to the agent. 

The next example demonstrates that when the principal does not directly benefit from randomizing actions, her primary channel to utilize the agent's information is by leveraging the agent's risk-tolerance.

 %This optimal payoff is determined by the mutual payoff function $\gamma_{\theta}(u)$, which distills the trade-offs between the principal's and the agent's objectives. The global shape of the principal's utility, $u_P(a, \theta)$, dictates her attitude towards risk—whether she benefits from randomizing her action without altering its expectation. In turn, the shape of the agent's utility, $u_A(a)$, determines the principal's ability to shift her expected actions without violating the agent's reporting incentives. Example \ref{exp:screening_intro} isolates this second effect, demonstrating how the principal can leverage the agent's preferences to her advantage.\footnote{The function $\gamma$ traces the utility possibility frontier, also referred to as the Pareto frontier.}  

\begin{example}[Audit]
To illustrate the optimal mechanism under full commitment, consider an audit problem. A risk-neutral audit company (P) decides on a score $a\in[0,1]$ it gives to a bank (A). The company's utility is $u_P(a,h) = a$ and $u_P(a,l) = 1-a$; she wants to give a high score only the high-standards bank ($\theta=h$). The bank, whose type is private, is risk-averse with utility $u_A(a) = \sqrt{a}$ and strictly prefers a higher score. The company's prior belief is $P(\theta=h)=\rho< 0.5$.

To apply Proposition \ref{prop:full_commitment} we identify the functions $\gamma_h(u)=u^2$ (strictly convex) and $\gamma_l(u)=1-u^2$ (strictly concave). The corresponding concave envelopes  are $\gamma_h^{co}(u)=u$ and $\gamma_l^{co}(u)=1-u^2$.Since $\gamma_h(u)$ has a non-trivial envelope, the company  optimally randomizes actions when the state is $h$. Similarly, since $\gamma_l(u)$ coincides with its envelope, the company optimally takes a deterministic action when the state is $l$. Therefore, the optimal audit policy must have the "firm-but-fair" form. It promises the bank an expected payoff of $\overline u$ regardless of the report\footnote{Maximizing $\rho \overline u+(1-\rho)(1-\overline u^2)$ with respect to $\overline u$ reveals that the optimal mechanism delivers $\overline u=\frac{\rho}{2(1-\rho)}$.}. If the bank reports being a high type, the audit company commits to the lottery:
\begin{equation*}
    \overline\pi_h:\begin{cases}
		a=1\quad \text{w.p.}\quad \overline u\\
		a=0 \quad \text{w.p.}\quad 1-\overline u
	\end{cases}
\end{equation*}
If it reports being a low type, the company commits to the deterministic action:
\begin{equation*}
    \overline\pi_l:\begin{cases}
		a=\overline u^2\quad \text{w.p.}\quad 1
	\end{cases}
\end{equation*}
On average, the company is more favorable to a bank that reports type $h$ (the expected action $\overline u$ is greater than $\overline u^2$). This higher reward, however, comes with a substantial risk—the possibility of a zero score. For the risk-averse bank, this risk exactly offsets the appeal of the higher expected score, making it indifferent between the two reports:
\begin{equation*}
	\mathbb E_{\overline \pi_h}[u_A(a)]=\overline u\cdot \sqrt 1+(1-\overline u)\cdot \sqrt0=\overline u=\mathbb E_{\overline \pi_l}[u_A(a)]
\end{equation*}

This carefully constructed risk-reward balance makes truth-telling incentive compatible. By leveraging the bank's risk aversion, the "firm-but-fair" policy secures the audit company a payoff of $\overline{V}=1-\rho+\frac{\rho^2}{4(1-\rho)}$, a strict improvement over its  no-communication baseline $V_B=1-\rho$.

	  	\label{exp:screening_intro}
\end{example}

 %\cmt{Add this line somewhere else} On top of that, shifting focus to the mutual payoffs reveals the channel through which the principal extracts benefit. 

In the following sections I move attention from optimality to worst-case optimality. As we will see, finding the worst-case optimal mechanism requires shifting focus from full commitment to partial. This shift qualitatively transforms the analysis. In particular, the agent-optimality constraint (\ref{eqm:ind}) is no longer solely a function of the committed plan $\pi$; it now additionally incorporates the principal's strategic response when she retains discretion:
$$\chi \mathbb E_{ \pi_m}[u_A(a)]+(1-\chi) \mathbb E_{ \sigma_P(m)}[u_A(a)]=const$$
In any equilibrium, the principal's discretionary strategy, $\sigma_P$, is a best response to the agent's messaging strategy, $\sigma_A$. This introduces a critical feedback loop: the agent's strategy now implicitly shapes his own incentive-compatibility constraint, a feature entirely absent under the full commitment case  (see (\ref{eqm:IC1})). This interdependence has two main implications for the analysis.

First, mechanisms optimized for full commitment become surprisingly fragile. Any plan $\{\pi_m\}_m$ designed to be incentive-compatible under full commitment will yield a securable payoff no better than the no-communication baseline once discretion is introduced (Lemma \ref{lem:chi_around_1}). The search for the worst-case optimal mechanism must therefore extend beyond this restrictive class of mechanisms.

Second, the move to worst-case implementation invalidates the standard focus on truth-telling equilibria. In a general mechanism $Q$, truth-telling is not necessarily a feasible strategy. On top of that, the principal must account for \emph{any} equilibrium the agent might select. Consequently, she loses the ability to fine-tune each action plan $\pi_m$ to optimally shape the trade-off $\gamma(a,\theta)$ for each specific state $\theta$.

%\cmt{Reposition to the introduction}

 %%%%%%%%%%%%%%%%%%%%%%%%%%%%%%% ------------IMPLEMENTATION-------------------
 
\section{Worst-case implementation}
\label{sec:worst-case_implementaation}

\subsection{Full commitment fails worst-case optimality}

This section establishes the limitations of mechanisms at the extremes of the commitment spectrum. While an optimal full-commitment mechanism can achieve the payoff $\overline V$ in a principal-preferred equilibrium, its worst-case performance is no better than receiving no information. Partial commitment can improve worst-case performance.

\begin{lemma}
Consider some mechanism $Q=(\{\pi_m\}_m,\chi)$. If $\chi\in \{0,1\}$, then $\inf_{E\in\mathcal{E}(Q)}V(Q,E)\le V_{B}$.
\label{lem:cheaptalk}
\end{lemma}

The proof follows from analyzing the two polar cases.

\textit{No commitment.} When $\chi=0$, the game reduces to cheap talk with transparent motives (\cite{lipnowski2020cheap}). Any such game possesses a babbling equilibrium where the agent's messages are uninformative. The principal's best response to uninformative messages is to disregard them, which results in her no-communication payoff, $V_B$.

\textit{Full commitment.} When $\chi=1$, any mechanism also admits a babbling equilibrium. The agent can always adopt a strategy in which his messages are uncorrelated with his type. The principal, bound by her commitment, must implement an action plan based on these uninformative messages, which again cannot yield a payoff greater than $V_B$.

The existence of such payoff-reducing equilibria is a structural feature of mechanisms with $\chi \in \{0,1\}$. In contrast, as demonstrated by Section \ref{exp:socrates}, a properly chosen mechanism with partial commitment can eliminate these equilibria and guarantee that only strictly informative ones survive. But even mechanisms with partial commitment cannot generally isolate a single equilibrium, as the next example demonstrates.

%This value, $V^\star$, represents the principal's maximin payoff\footnote{This objective spiritually corresponds to $\nu$-generalized regret with $\nu=0$ (See \cite{robust_bergemann2023}).} and is necessarily bounded by the benchmark: $V^\star\le \overline{V}$. The following example demonstrates that this inequality can be strict, which occurs when any mechanism admitting a truth-telling equilibrium also admits a less informative one.

\begin{example}
\label{exp:multiple}
	
To illustrate this, we can extend the setting of Example \ref{exp:screening_intro}. Let the set of bank types be $\Theta=\{l,i,h\}$, with prior probabilities $\rho_l = 0.5$ and $\rho_i=\rho_h=0.25$. The analysis is restricted to mechanisms $Q$ that admit a truth-telling equilibrium, supported by a strategy $\sigma^\tau$. The message space is $\mathcal{M}=\{m_l,m_i,m_h\}$, ordered without loss of generality such that the bank's payoffs under commitment are non-decreasing: $u_{m_l}^\pi\le u_{m_i}^\pi\le u_{m_h}^\pi$.
	
	The following lemma establishes that any such mechanism also admits another, less informative equilibrium strategy.
	
	\begin{lemma}
		For any mechanism $Q$ that admits a truth-telling strategy $\sigma^\tau\in\Sigma(Q)$, there exists a feasible strategy $\sigma_A\in\Sigma(Q)$ such that:

	\begin{table}[h!]
	\begin{tabular}{cccc}
		$\sigma_A(m|h)$  = & ($\overset{m_l}{0}$, & $\overset{m_i}{1}$, & $\overset{m_h}{0}$)\\
		\\
		$\sigma_A(m|i)$	=	&	($0$,	&$1$,	&	$0$)\\
		\\
		$\sigma_A(m|l)$	=	&($0$,	&$0.5$,	&$0.5$)
	\end{tabular}
	\end{table}

	\label{lem:multiple_equilibria}
	\end{lemma}
	Under the strategy $\sigma_A$, types $h$ and $i$ pool on the intermediate message $m_i$, while the low type $l$ randomizes between messages $m_i$ and $m_h$. The fact that only the low type ever reports the high message $m_h$ is reminiscent of the Socrates Effect (Lemma \ref{lem:socrates}). This strategy is strictly less informative (in the Blackwell sense) than the fully revealing strategy $\sigma^\tau$.

The existence of this alternative feasible strategy explains why the principal generally cannot worst-case implement the benchmark payoff $\overline V$. Since, by Lemma \ref{lem:multiple_equilibria}, the strategy $\sigma_A$ is available to the agent in any mechanism that admits truth-telling, the principal's guaranteed payoff is bounded by the outcome under $\sigma_A$. To worst-case implement $\overline V$, the principal would need to achieve this payoff even when the agent employs the less informative $\sigma_A$. This is generally impossible, as the benchmark mechanism typically requires the fine-grained information revealed only by $\sigma^\tau$.

The logic underpinning Lemma \ref{lem:multiple_equilibria} is that the strategy $\sigma_A$ is constructed to hold the agent's equilibrium payoffs constant. The beliefs induced by the on-path messages under $\sigma_A$ are such that the uncommitted principal's best responses, and thus the agent's payoffs from facing the uncommitted principal, are identical to those under the truth-telling strategy $\sigma^\tau$. Because the agent's total expected payoff $U_A(m)$ is a convex combination of payoffs from the committed and uncommitted principal, and the indifference condition holds for $\sigma^\tau$, it must also hold for $\sigma_A$.\footnote{To show that $\sigma_A$ is feasible,  the proof also establishes that when the agent plays $\sigma_A$, he has no incentives to deviate to the off-path $m_l$.} This argument is independent of the specific committed plan $\pi$, and therefore applies to any mechanism admitting a truth-telling equilibrium.
	\end{example}

The multiplicity of equilibria identified in Example \ref{exp:multiple} is a structural feature of mechanisms that commit to distributions over actions. It is useful to distinguish between two sources of equilibrium multiplicity in general. The first arises when distinct equilibria induce different distributions of actions from the uncommitted principal. The second, more persistent form occurs when distinct equilibria induce different posterior beliefs but the \emph{same} distribution of uncommitted actions.

Partial commitment can resolve the first type of multiplicity by altering the agent's committed payoffs across these different action distributions. It cannot, however, resolve the second. The latter form persists because the agent's incentive-compatibility constraint is a function of the principal's overall action distribution. If multiple distinct messaging strategies generate the same best-response actions from the uncommitted principal, then any committed plan that renders one of these strategies incentive-compatible must also do so for the others. Fundamentally, this strategic equivalence stems from the mapping from the high-dimensional space of posterior beliefs to the lower-dimensional space of actions.

%The next result generalizes this insight. It establishes that the existence of such alternative, less-informative equilibria places a sharp bound on the amount of information the principal can guarantee. Specifically, the maximum information that can be secured in the worst case is that which can be conveyed through a binary message structure. To state this result we need to introduce two genericity assumptions regarding the players' preferences and the strategic environment.

\subsection{Characterization}

This section characterizes the worst-case optimal mechanism under three assumptions. 

\begin{assumption}[Genericity 1]
	For any $\mu\in\Delta(\Theta)$, the function $\sum_{\theta\in\Theta}\mu(\theta)u_P(a,\theta)$ has at most two maximizers. Moreover, it has a unique maximizer for all $\mu$ outside a set of Lebesgue measure zero. Each $u_P(\theta)$ has a unique maximizer. 
	\label{assp:genericity_1}
\end{assumption}

The first assumption requires that the principal has a unique optimal action for almost every posterior belief. This assumption provides a sufficient condition for a key intermediate result: if a mechanism admits any non-trivial equilibrium, it must also admit a two-message equilibrium. The main conclusions of the paper are robust to the relaxation of this assumption, which is imposed primarily for analytical convenience.%\footnote{An alternative, more direct assumption requires that if under an information structure the uncommitted principal delivers the agent payoffs $u_1,...,u_M$, then there exists a binary information structure that delivers the agent $u_i$ and $u_j$ for some distinct $i,j\le M$.}

For a \emph{two-message strategy }$\sigma_A$ let 
\[\Delta(\sigma_A)=\sup_{\sigma_P\textit{ - BR to }\sigma_A}\left |\mathbb E_{\sigma_P(m_1)}[u_A(a)]-\mathbb E_{\sigma_P(m_2)}[u_A(a)]\right |\]
 denote the largest payoff spread that strategy $\sigma_A$ can induce from the uncommitted principal. The set of \emph{polarizing messaging strategies}, $\Sigma^\star$, is defined as the set of strategies that achieve the maximum possible spread:
\begin{align}
	\Sigma^\star	=\argmax_{\sigma_A-\text{two message strategy}} \Delta(\sigma_A)	\notag
\end{align}
Let $\overline \Delta = \max_{\sigma}\Delta(\sigma)$ denote this maximal spread. The quantity $\Delta(\sigma_A)$ can be interpreted as a measure of the strategic impact of the agent's message, and a polarizing strategy is one that maximizes this impact.\footnote{This difference is related to the informativeness of the agent's messaging, especially when the principal faces a monotone decision problem. See, for example, \cite{athey1998value}.} It is also useful to introduce mechanisms that make polarizing messaging incentive-compatible.

\begin{definition}[Polarizing Mechanism] A \emph{polarizing mechanism} $Q$ is a two-message mechanism with partial commitment ($\chi<1$) that admits a polarizing messaging strategy in equilibrium: $\Sigma^\star\cap \Sigma(Q)\neq \varnothing$.
\end{definition}

\noindent As established in Lemma \ref{lem:pol_mech_properties}, polarizing mechanisms exist and can be simply characterized. I now characterize the worst-case optimal mechanism under the following two assumptions.

\begin{assumption}[Genericity 2]
Let $(u,v)$ be a pair of payoffs to the agent resulting from the principal's best responses to the two messages of some binary information structure $\sigma$. Let $\mathcal U$ be the set of all such Bayes-plausible payoff pairs. The set $\mathcal U$ has a unique polarizing point $(u^\star,v^\star)$ (up to relabeling of the entries) that solves:
\[(u^\star,v^\star)\in \argmax_{(u,v)\in \mathcal U}|u-v|\]
\label{assp:genericity_2}
\end{assumption}

Assumption \ref{assp:genericity_2} is generically satisfied when the strategic players are sensitive to their beliefs.\footnote{The set $\mathcal U$ is compact and varies continuously with the primitives $u_A$ and $u_P$. If the primitives happen to induce a non-generic case, an arbitrarily small perturbation will almost surely restore the uniqueness property.} To see why this assumption is generically satisfied, think of all payoff pairs the agent can induce from the uncommitted principal via binary communication. %Assumption \ref{assp:genericity_2} requires that if some binary information structure maximally spreads the agent's payoffs, then the agent's payoffs mu. 
This assumption is violated if the agent has two messaging strategies, that happen to generate two \emph{distinct} payoff pairs for the agent, that both achieve the same, \emph{maximal} payoff difference. 

The following assumption provides a sufficient condition for the subsequent results by restricting the analysis to environments where credible communication requires commitment. While stronger than necessary, it serves to simplify the exposition.

\begin{assumption}
The agent cannot improve his payoff via pure cheap talk. That is, in any cheap talk equilibrium the agent's payoff does not exceed his payoff under no communication.
\label{assp:no_cheap_talk}
\end{assumption}

%The agent who uses two messages cannot benefit from both of them (compared to the prior). Formally, for any binary messaging strategy $\sigma_A$ the agent's uncommitted payoffs  $u_A^{\sigma_{P}(\sigma_A)}(m)\le u_\rho$ for some on-path $m$.

Assumption \ref{assp:no_cheap_talk} is a non-trivial restriction in environments with three or more states ($n \ge 3$) and ensures that the principal cannot extract decision-relevant information without leveraging some degree of commitment.\footnote{This is equivalent to requiring that the the agent's payoff as a function the principal's posterior beliefs coincides with its quasi-concave envelope at $\rho$ (\cite{lipnowski2020cheap}).}$^,$\footnote{A possible practical interpretation is that if the agent can strictly benefit from revealing information via cheap talk, he would have done so before the principal resorts to commitment to extract additional information.} The analytical framework developed in this paper can be extended to settings where this assumption is relaxed.\footnote{In such an environment, the principal's problem is augmented by a trade-off between leveraging information that is credibly revealed via a cheap talk equilibrium and eliciting information through the commitment component of the mechanism. I conjecture that a modified version of Theorem \ref{thm:carrot_stick} exists for this case.}

\begin{theorem}[Carrot-Stick Principle]
Suppose Assumptions \ref{assp:genericity_2} and \ref{assp:no_cheap_talk} hold. Then the (virtual) worst-case optimal mechanism exists. Moreover, to worst-case implement $V^\star$ it is sufficient to focus on polarizing mechanisms:
\[V^\star=\sup_{Q - polarizing}\inf_{E\in \mathcal E(Q)}V(Q,E)\]	
\label{thm:carrot_stick}
\end{theorem}

Theorem \ref{thm:carrot_stick} reveals a sharp contrast between standard and worst-case optimality: whereas the optimal mechanism uses $n$
 messages (Section \ref{sec:benchmark}), the worst-case optimal mechanism uses only two. This finding underscores the insufficiency of the standard revelation principle as a tool for worst-case implementation. Although the revelation principle guarantees that a full-commitment mechanism can achieve the worst-case payoff in some of its equilibria, it offers no guidance as to which of these equilibria can be implemented as a worst case of a partial commitment mechanism. Proposition \ref{prop:virtual_revelation} establishes a version of the revelation principle, useful to my setting.

\subsection{Proof Sketch}
The analysis of the worst-case equilibrium in two- versus many-message mechanisms involves two countervailing effects. First, allowing more messages expands the set of equilibria, which may worsen the worst-case outcome. Second, for any fixed pair of on-path messages, a richer message space can restrict the set of supportable two-message equilibria. Specifically, off-path messages $\tilde m$ create an outside option for the agent, so any equilibrium must yield him a payoff exceeding $\chi \mathbb E_{\pi_m}[u_A(a)]$.

The logic of the proof can be illustrated by first considering the simplified case where the set of polarizing strategies $\Sigma^\star$ is a singleton, $\{\sigma^\star\}$. For an arbitrary mechanism $Q$, it can be shown that there exists an equilibrium strategy $\sigma_A \in \Sigma(Q)$ that is Blackwell-dominated by $\sigma^\star$. This establishes that the polarizing strategy is more informative than some equilibrium strategy of any given mechanism. Lemma \ref{lem:pol_mech_properties} shows that under any polarizing mechanism $Q'$, the agent's unique equilibrium strategy is $\sigma^\star$. Because $\sigma^\star$ is more informative than $\sigma_A$, a polarizing mechanism $Q'$ can be constructed to yield a payoff at least as high as that obtained under $(Q, \sigma_A)$.

This direct argument does not extend to the general case where $\Sigma^\star$ is not a singleton. While it remains true that for any mechanism $Q$ there exists $\sigma_A \in \Sigma(Q)$ and $\sigma^\star \in \Sigma^\star$ such that $\sigma^\star$ Blackwell-dominates $\sigma_A$, one cannot guarantee that this particular $\sigma^\star$ corresponds to the worst-case equilibrium in a polarizing mechanism designed to exploit it, as it was in the singleton $\Sigma^\star$ case. Another polarizing strategy, $\tilde{\sigma} \in \Sigma^\star$, could yield a lower payoff in $Q'$ than $\sigma_A$ does in $Q$.

The full proof proceeds in several steps. First, Assumption \ref{assp:genericity_2} is used to establish the convexity of the set of polarizing strategies, $\Sigma^\star$. This property allows for the application of Kakutani's fixed point theorem to a properly constructed correspondence from $\Sigma^\star$ to itself. The existence of a fixed point then guarantees that for any arbitrary mechanism $Q$, there exists a polarizing mechanism whose worst-case payoff is no lower than the payoff from some equilibrium of $Q$.

The proof for the existence of a worst-case optimal mechanism follows a similar logic. We can invoke the Eilenberg-Montgomery fixed point theorem to show that the correspondence 

\[\argmax_{Q-\text{(\ref{eqm:IC1wc})}}\min_{\sigma\in\Sigma^\star}V(Q,\sigma)\] 
that maps the set of polarizing mechanisms on itself has a fixed point. The main crux in proving this is verifying that this correspondence has contractible values. This can be shown by using results from algebraic topology. 

\qed

\subsection{Geometric characterization}

%\begin{proposition}[yes-no principle] Suppose Assumptions \ref{assp:genericity_2} and \ref{assp:no_cheap_talk} hold. To worst-case implement $V^\star$ it is sufficient to focus on two-message mechanisms:
%\[V^\star=\sup_{Q-\textit{two message}}\inf_{E\in\mathcal E(Q)}V(Q,E) \]
%\label{prop:yes-no}	
%\end{proposition}

 %In the remainder of this section, the analysis is restricted to a binary message space, $\mathcal M=\{yes,no\}$. The characterization of the worst-case optimal mechanism requires identifying a specific class of agent strategies: those that are polarizing. A polarizing strategy is one that maximizes the difference in the agent's expected payoffs from the uncommitted principal's best-response actions.

Theorem \ref{thm:carrot_stick} allows us to restrict attention to two messages: $\mathcal M\equiv \{yes,no\}$. The following lemma provides two properties of polarizing mechanisms that are central to the characterization of the worst-case optimal payoff.

\begin{lemma}
	1. A two-message mechanism $Q=((\pi_{yes},\pi_{no}),\chi)$ with $\chi<1$ is polarizing if and only if it satisfies the condition:
	\[\frac{\chi}{1-\chi}\left|\mathbb E_{\pi_{yes}}[u_A^\pi(a)]-\mathbb E_{\pi_{no}}[u_A^\pi(a)]\right|=\overline\Delta\]
	2. For any polarizing mechanism $Q$, the set of equilibrium messaging strategies is precisely the set of polarizing strategies: $\Sigma(Q) = \Sigma^\star$.
	\label{lem:pol_mech_properties}
\end{lemma}

Part 1 of the lemma provides a simple equality that defines the entire class of polarizing mechanisms (analogous to (\ref{eqm:IC1}) for direct mechanisms). For a given level of commitment $\chi$, the committed plan $(\pi_{yes},\pi_{no})$ must be designed such that the incentives it provides, weighted by $\chi$, exactly balance the maximal possible incentive spread from the uncommitted principal, $\overline\Delta$, weighted by $1-\chi$. 

Part 2 establishes that polarizing mechanisms solve the equilibrium selection problem for the worst-case concerned principal. They restrict the agent's equilibrium behavior exclusively to the set of polarizing strategies, $\Sigma^\star$. This property, combined with Theorem \ref{thm:carrot_stick} and Proposition \ref{prop:virtual_revelation}, allows the general problem of finding $V^\star$ to be recast into a more tractable form. Let $V(\pi,\sigma_A)$ be the principal's committed payoff from action plans $\pi$ when the agent's strategy is $\sigma_A$.

\begin{proposition}
Suppose Assumptions \ref{assp:genericity_2} and \ref{assp:no_cheap_talk} hold. The worst-case optimal payoff $V^\star$ is the solution to the following optimization problem:
 	\[V^\star=\sup_{\pi-(\text{\ref{eqm:IC1wc}})}\inf_{\sigma\in\Sigma^\star} V(\pi,\sigma)\]
 	where 
 	\[\mathbb E_{\pi_{yes}}[u_A(a)]=\mathbb E_{\pi_{yes}}[u_A(a)]\quad \label{eqm:IC1wc}\tag{AO$_1^\star$}\]
\label{prop:wcI_full_commitment}
\end{proposition}

This result reflects the virtual nature of worst-case implementation and reduces the search for the worst-case optimal payoff to an optimization over full-commitment, incentive-compatible, two-message mechanisms, where the principal anticipates that the agent will adversarially select the worst strategy from the fixed set $\Sigma^\star$. This simplification follows from the general principle (see Section  \ref{sec:virtual_revelation}): if a class of mechanisms restricts the set of equilibrium strategies to $\Sigma^\star$, then the set of supportable outcomes is equivalent to the set of outcomes from full-commitment mechanisms where the agent is assumed to select a strategy from $\Sigma^\star$. 

The primary feature distinguishing the optimization in Proposition \ref{prop:wcI_full_commitment} from the benchmark case is the inner minimization over the set $\Sigma^\star$, which is exogenous to the mechanism's design. This minimization is necessary because, even under Assumption \ref{assp:genericity_2}, the set of polarizing strategies may not be a singleton if distinct messaging strategies induce identical best-response payoffs from the uncommitted principal. To facilitate a sharp characterization of the optimal mechanism, the analysis now proceeds under the simplifying assumption that the polarizing strategy is unique.

\begin{assumption}
	$\Sigma^\star$ is a singleton: $\Sigma^\star\equiv \{\sigma^\star\}$.
	\label{assp:unique}
\end{assumption}

Under this assumption, the agent's equilibrium strategy in any polarizing mechanism is uniquely determined to be $\sigma^\star$. This simplifies the principal's problem as she now takes $\sigma^\star$ as given.  This reduces her ex-ante uncertainty from the multi-dimensional state space $\Theta$ to the binary message space $\{yes, no\}$ with a marginal distribution $P^\star(m)=\sum_{\theta\in\Theta}\sigma^\star(m|\theta)\rho(\theta)$. We can therefore define the principal's \emph{polarized utility} for each message $m$ as the expected mutual payoff conditional on receiving that message: $\gamma_{m}(u) = \sum_\theta \mu^\star(\theta|m)\gamma_{\theta}(u)$, where the posterior $\mu^\star(\theta|m)$ is derived from the prior and the agent's strategy $\sigma^\star$. The next result characterizes the worst-case optimal mechanism in terms of concave envelopes of these polarized utilities.

\begin{theorem} Suppose Assumptions \ref{assp:no_cheap_talk} and  \ref{assp:unique} hold. Then
\[V^\star=\max_{u\in[0,1]} P^\star(yes)\gamma_{yes}^{co}(u) + P^\star(no)\gamma_{no}^{co}(u)\]
\label{thm:characterization_general}
\end{theorem}

 This result clarifies the role of partial commitment in this framework: its primary function is not to alter the structure of optimal payoffs, but rather to discipline the agent's equilibrium behavior by restricting the set of feasible messaging strategies to $\Sigma(Q) = \{\sigma^\star\}$. Once this restriction is imposed, the optimal payoff can be determined using techniques analogous to the full commitment benchmark—namely, by optimizing over the agent's induced payoff $u$. 
 
\begin{example}
\label{exp:worst-case}

	To illustrate the characterization in Theorem \ref{thm:characterization_general}, consider an environment with three states, $\Theta\equiv\{l,m,h\}$, and a uniform prior. Imagine a health insurance company that faces an applicant, whose type (e.g. risk) is captured by $\theta$. The policy premium is $a\in[0,1]$ and the insurance company wants to set a higher premium to the lower-type agent, who is riskier. The mutual payoff function $\gamma(u, \theta)$ is designed to capture this preference. At each $\theta$ it is strictly single-peaked at the company's optimal policy premium. Specifically, $\gamma_{\theta}(u)$ for each state is defined as the sum of a quadratic term and a bump function, where $Bump(x,b) = e^{b^2/(x^2-b^2)}$ for $|x|\le b$ and zero other\emph{wise}. The specific forms are chosen to satisfy Assumption \ref{assp:unique}, be tractable for computations and are:
	
	\[\gamma(u,l)=0.02 \left(1 -\frac{ (u - 0.1)^2}{0.81} \right)+Bump(u-0.1,0.4)\]
	\[\gamma_{m}(u)=0.02 \left(1 -\frac{ (u - 0.5)^2}{0.25}\right)+Bump(u-0.5,0.5)\]
	\[\gamma(u,h)=0.02\left (1 - \frac{(u - 0.9)^2}{0.81}\right)+Bump(u-0.9,0.6)\]
These functions are constructed such that the applicant's payoff from the company's optimal premium against the low type is $u_{l}=u_A(a(l))=0.1$.  Similarly, $u_{m}=0.5$, and $u_{h}=0.9$. In other words, the low type (riskiest) receives the largest premium, while the high type (safest) receives the lowest premium.

In this setting, the unique polarizing strategy $\sigma^\star$ separates the type $h$ from the types $\{l,m\}$. This strategy induces the following polarized utility functions:
\[\gamma(u,yes)=\gamma(u,h)\]
\[\gamma(u,no)=\frac{1}{2}\left(\gamma(u,l)+\gamma_{m}(u)\right)\]
Both $\gamma(u,yes)$ and $\gamma(u,no)$ have concave hulls, depicted in Figure \ref{fig:polarizing}.

	\begin{figure}[h!]
		\includegraphics[scale = .6]{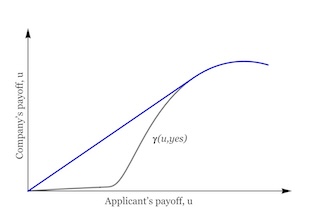}
		\includegraphics[scale = .6]{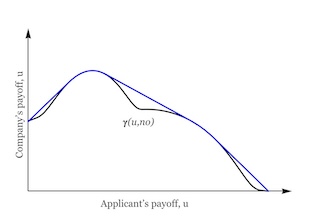}
		\caption{Blue - effective utility functions, red - their concave hulls. $\gamma(u,yes)$ is on the left, $\gamma(u,no)$ is on the right}
		\label{fig:polarizing}
	\end{figure}

The numerical solution to this example, summarized in Table \ref{tab:exp}, highlights the quantitative differences between the  worst-case optimal and optimal mechanisms\footnote{In the case of this example, there is no need to verify Assumption \ref{assp:no_cheap_talk}. The polarizing strategy delivers the agent the largest possible uncommitted payoff, $0.9$, to one of the messages, which turns out to be a sufficient condition for  polarizing mechanisms to be optimal.}. Compared to the no-communication baseline payoff of $V_B\approx0.21$, the worst-case optimal mechanism yields $V^\star\approx 0.23$, an increase of approximately 10\%. The benchmark full-commitment mechanism achieves a significantly higher payoff of $\overline{V}\approx 0.27$, a 30\% improvement over the baseline. Crucially, the benchmark, which requires non-degenerate lotteries for types $l$ and $h$ to ensure incentive compatibility, cannot be worst-case implemented.

	\begin{table}[h!]
		\begin{tabular}{l|ccc}
					& No communication & Worst-case communication & Full commitment \\
					\hline\\
					Ins's question & -- & \emph{Is your type $h$?} & \emph{What is your type?}\\
					\\
			Ins's payoff & $V_B\approx0.21$ & $V^\star\approx 0.23$ & $\overline{V}\approx 0.27$ \\
			%\hline \\
			\\
			A's payoff & $u_B\approx 0.67$ & $u^\star\approx 0.3$ & $\overline u\approx 0.51$ \\
		\end{tabular}
		\caption{Comparison of outcomes under alternative mechanisms.}
		\label{tab:exp}
	\end{table}
	
\end{example}

This section concludes by demonstrating a special case where this trade-off is resolved and the worst-case optimal payoff coincides with the full-commitment benchmark: when the agent's private information is binary.

\begin{proposition}
 Suppose $n=2$ and Assumption \ref{assp:unique} holds. Then $V^\star=\overline V$.	 \label{prop:n=2}
\end{proposition}

When the agent's information is binary, by keeping discretion, the principal can ensure that the agent reveals all his private information. The proof follows from the confluence of several results. Theorem \ref{thm:carrot_stick} establishes that two messages are sufficient to implement $V^\star$. When $n=2$, a two-message structure also has sufficient capacity to fully reveal the agent's private information. A truth-telling strategy induces the extreme actions $a(\theta_1)=0$ and $a(\theta_2)=1$, making it a polarizing strategy. Assumption \ref{assp:unique} then ensures this truth-telling strategy is the unique element of $\Sigma^\star$ (recall Lemma \ref{lem:socrates}). Consequently, any polarizing mechanism restricts the set of equilibrium strategies to truth-telling. The problem of finding $V^\star$ then reduces to finding the optimal full-commitment payoff under truth-telling, which, by Proposition \ref{prop:full_commitment}, is $\overline V$. As discussed before, this equivalence between $V^\star$ and $\overline V$ breaks down when $n>2$.

\subsection{Polarizing messaging}

While Theorem \ref{thm:characterization_general} provides a method for calculating the worst-case optimal payoff, it does not fully elucidate the qualitative structure of the communication $\sigma^\star$ it induces. As the next result shows, when unique, the polarizing strategy takes on a specific form.

\begin{proposition}[Separation]
There exists a polarizing messaging strategy $\sigma^\star$ with the following structure: there exist states $\theta^\star,\theta^+\in\Theta$ such that
 \[supp(P(\cdot|yes))\subseteq\{\theta^\star,\theta^+\}\]
 and \[supp(P(\cdot|no))=\Theta\setminus\{\theta^\star\}\]
\label{prop:separation}
\end{proposition}

When combined with Assumption \ref{assp:unique}, which ensures the polarizing strategy is unique, Proposition \ref{prop:separation} reveals the canonical structure of communication in the worst-case optimal mechanism. The state space is effectively partitioned into three categories: a \emph{target state} $\theta^\star$, a \emph{link state} $\theta^+$, and a set of \emph{pooled states} $\Theta\setminus\{\theta^\star,\theta^+\}$.

The communication protocol is designed to distinguish the target state from the pooled states. The message "yes" serves as a confirmation, revealing that the true state is either the target or the link state, while ruling out all pooled states. The message "no" serves as a denial, ruling out the target state $\theta^\star$ but revealing little information about the remaining states. The link state $\theta^+$ is unique in that it is not ruled out by either message (as seen in Example \ref{exp:multiple}). In certain environments, as seen in Examples \ref{exp:screening_intro}, \ref{exp:worst-case}, the link state may be absent. In such cases, the communication protocol reduces to a direct and truthful inquiry about whether the state is $\theta^\star$.

\begin{example}

To illustrate the case involving a link state, we can modify the environment of Example \ref{exp:worst-case} by specifying a strictly concave mutual payoff function in the high type case: $\gamma(u,h)=0.5-(1-u)^2$. This change removes the insurance company's incentive to randomize the premium when the applicant is known to be $h$.

In this modified setting, the target state is $\theta^\star=h$, as it was in Example \ref{exp:worst-case}, but a link state $\theta^+=m$ now emerges. The new unique polarizing strategy $\sigma^\star$ is given by:
\[\sigma^\star=\begin{pmatrix}
		0&0.155&1\\
		1&0.845&0
	\end{pmatrix}\]
Under this strategy, the intermediate type ($\theta=m$) does not report truthfully in response to the effective query "Is your type $h$?", instead mixing between the two messages. The insurance company, however, anticipates this reporting behavior and forms its posterior beliefs accordingly, allowing itself to correctly interpret the applicant's signals.\footnote{Although the company may have multiple best responses to the applicant's messages, the equilibrium is uniquely determined by the condition that the induced payoff spread equals $\overline\Delta$.} The presence of the link state implies that the polarized utility functions are mixtures and do not coincide with the mutual payoff function of any single state.

This case presents two notable contrasts with the previous example. First, as shown in Table \ref{tab:exp2}, the worst-case optimal mechanism delivers a larger expected utility to the agent than the benchmark full-commitment mechanism does. Second, the structure of the worst-case optimal plan is inverted: randomization is now required conditional on the denial message ("no"), rather than the confirmation message ("yes").

\begin{table}[h!]
	\centering
	\begin{tabular}{l|ccc}
				& No communication & Worst-case communication & Full commitment \\
				\hline\\
				P's question & -- & \emph{Is your type $h$?} & \emph{What is your type?}\\
				\\
		P's payoff & $V_B\approx0.238$ & $V^\star\approx 0.24$ & $\overline{V}\approx 0.29$ \\
		\\
		A's payoff & $u_p\approx 0.69$ & $u_P\approx 0.63$ & $u_C\approx 0.61$ \\
	\end{tabular}
	\caption{Comparison of outcomes in the modified environment of Example \ref{exp:worst-case_2}.}
	\label{tab:exp2}
\end{table}

	\label{exp:worst-case_2}
\end{example}

The analysis in Examples \ref{exp:worst-case} and \ref{exp:worst-case_2} illustrates the two forms the worst-case optimal mechanism can take. The structure of the mechanism is determined by a trade-off in the selection of the target state, $\theta^\star$. The principal may identify a target state where randomization is particularly valuable (i.e., where $\overline\gamma_{\theta}(u) > \gamma_{\theta}(u)$). In this case, the optimal mechanism elicits a confirmation of this state and employs randomization conditional on that confirmation. Alternatively, the principal may identify a target state where randomization offers no benefit (i.e., where $\gamma_{\theta}(u)$ is concave) and separate it from the remaining states. In this case, randomization is optimally employed conditional on the denial message.

\section{Discussion}

Interpreting the model in the context of algorithmic decision-making, where $\chi$ represents the fraction of decisions handled by an algorithm and $1-\chi$ the fraction handled by a strategic human, the analysis highlights several caveats for a principal concerned with worst-case outcomes.

First, full automation ($\chi=1$) is never worst-case optimal. The strategic discipline required to induce informative reporting stems entirely from the agent's anticipation of facing a discretionary human decision-maker. While the level of automation can be arbitrarily close to one without sacrificing optimality, a non-zero probability of human intervention is essential. Second, the amount of information that can be secured is limited to that which can be conveyed by a binary signal. Providing a richer message space expands the agent's opportunities for strategic manipulation, thereby lowering the principal's worst-case payoff.

Third, for any non-negligible level of human discretion, the committed plan of the algorithm must provide countervailing incentives to the anticipated response of the human. If both the algorithm and the human are expected to reward the same message, the agent's indifference condition will be violated, leading to uninformative pooling on the high-payoff message. To maintain informative equilibria, the algorithm's committed plan must, in expectation, counteract the human's strategic response. This may require the algorithm to respond to a "good" signal with a "bad" action, and vice-versa (recall the Socrates Effect).

While this ex-post suboptimality appears costly, its magnitude diminishes as automation becomes extensive. As $\chi$ approaches one, the required difference between the algorithm's committed actions needed to offset the human's influence becomes arbitrarily small, thus minimizing the associated inefficiency. Crucially, however, the \emph{threat} of human intervention remains, ensuring that the agent's report is informative even in the limit. The algorithm can then capitalize on this guaranteed informativeness to achieve a payoff strictly above the no-communication baseline.

\section{Appendix.}

\subsection{Discontinuity of the worst-case payoffs. }
Let $V_{CT}$ denote the principal's highest achievable payoff in a cheap talk game (i.e., when $\chi=0$). The following lemma establishes that as a mechanism's commitment level approaches one, its performance guarantee degrades to that of a cheap talk game.

\begin{lemma}
	Fix any plan $\pi$. For any sequence of mechanisms $Q_n=(\pi,\chi_n)$ with $\chi_n\rightarrow1$, the highest attainable equilibrium payoff is bounded by the cheap talk payoff:
\[\lim_{n\rightarrow \infty}\sup_{E\in\mathcal E(Q_n)}V(Q_n,E)\le V_{CT}\]
\label{lem:chi_around_1}
\end{lemma}

The proof considers two cases based on the properties of the committed plan $\pi$.

First, suppose the committed plan $\pi$ does not satisfy the full-commitment agent-optimality condition (\ref{eqm:IC1}). In this case, there exists a message $m_0$ that yields the agent a strictly higher payoff under the committed plan than some other message. As $\chi_n \rightarrow 1$, the agent's utility becomes dominated by the payoff from the committed plan. For $n$ sufficiently large, the agent of any type will strictly prefer to send message $m_0$. This precludes any non-trivial equilibrium, collapsing communication and yielding the principal a payoff no greater than $V_{CT}$.

Second, suppose the committed plan $\pi$ satisfies (\ref{eqm:IC1}), making the agent indifferent among all messages under the committed plan. In this scenario, the agent's choice of message is determined entirely by the strategic component of his utility, which depends on the uncommitted principal's best response. The strategic interaction is therefore equivalent to a cheap talk game played over the set of messages for which the agent is indifferent. Consequently, the principal's payoff in any equilibrium is bounded above by $V_{CT}$.

\subsection{Virtual revelation}

\label{sec:virtual_revelation}

This section establishes a useful, albeit weaker, version of the revelation principle, adapted for the analysis of partial commitment and the sets of outcomes it can support. To formalize the analysis, let the set of equilibrium outcome pairs for a mechanism $Q$ be defined as:
\[w(Q)=\{(u,v)\mid \exists (\sigma_A,\sigma_P)\in \mathcal E(Q)\textit{ s.t. }v=V(Q,(\sigma_A,\sigma_P))\text{ and }\]\[u=\chi \mathbb E_{ \pi_m}[u_A(a)]+(1-\chi)\mathbb E_{ \sigma_P(m)}[u_A(a)]\textit{ for on-path m}\}\]
This set includes all agent and principal payoffs achievable in any equilibrium of $Q$. Similarly, for a given set of messaging strategies $\Sigma$ and a full-commitment action plan $\pi$, let the set of achievable outcomes be:
\[w_{fc}(\pi,\Sigma)=\{(u,v)|\exists \sigma_A\in \Sigma\textit{ s.t. }v=V(\pi,\sigma_A)\text{ and }u= \mathbb E_{ \pi_m}[u_A(a)]\textit{ for on-path m}\}\]
The following proposition connects these two sets of outcomes.

\begin{proposition}
 
 Suppose $\Sigma = \Sigma(Q')$ for some mechanism $Q'$. Then for any full-commitment action plan $\pi$, there exists a sequence of mechanisms $Q_n$ such that $\lim_n w(Q_n) = w_{fc}(\pi,\Sigma)$.   Moreover, if $\Sigma(Q)\equiv \Sigma^\star$ then there exists $\pi$ such that $w(Q)=w_{fc}(\pi,\Sigma^\star)$.
 \label{prop:virtual_revelation}
\end{proposition}

Proposition \ref{prop:virtual_revelation} establishes that if a set of equilibrium strategies $\Sigma$ is supportable by some mechanism, then the set of outcomes from any full-commitment mechanism restricted to $\Sigma$ can also be generated by a partial commitment mechanism. For the specific class of polarizing mechanisms, this relationship is an equivalence: the set of equilibrium outcomes of any polarizing mechanism is precisely the set of outcomes achievable under a full-commitment mechanism where the agent's strategy is restricted to the set of polarizing strategies, $\Sigma^\star$.

This equivalence, combined with the result that the search for the worst-case optimal payoff can be restricted to polarizing mechanisms (Theorem \ref{thm:carrot_stick}), provides the foundation for the characterization in Proposition \ref{prop:wcI_full_commitment}. It establishes that $V^\star$ can be found by solving the simpler problem of a principal who fully commits to an incentive-compatible plan, anticipating that the agent will adversarially select a strategy from the fixed set $\Sigma^\star$.

\begin{proof}
The proof relies on the observation that the agent's indifference conditions, which shape the set of equilibria $\mathcal{E}(Q)$, can be preserved under adjustments to the commitment level $\chi$ and the committed action plan $\pi$. For any mechanism $Q$, it is possible to construct a new partial commitment mechanism with a different $\chi$ and a modified action plan $\pi'$ such that the set of equilibrium strategies remains $\Sigma(Q)$. By appropriately choosing these modifications, the outcomes of any full-commitment mechanism can be (virtually) replicated while holding the set of feasible strategies fixed.
\end{proof}

A direct corollary of this principle is that any payoff achievable under full commitment is also achievable as the outcome of some equilibrium of a partial commitment mechanism. This includes the benchmark payoff $\overline V$.

\begin{proposition}
	$$\overline V=\sup_{Q:\chi<1}\sup_{E\in\mathcal E(Q)}V(Q,E)$$
\label{prop:implementability_0}
\end{proposition}

Proposition \ref{prop:implementability_0} reveals that even if we view full commitment as practically implausible, restricting attention to mechanisms with strictly partial commitment does not change the principal's range of attainable payoffs.

Proposition \ref{prop:implementability_0} also formalizes the observation from Lemma \ref{lem:multiple_equilibria}. It demonstrates that while partial commitment can achieve the benchmark payoff in a principal-preferred equilibrium, generally it can never guarantee it. The equivalence of the supremum payoffs for $\chi<1$ and $\chi=1$, combined with Theorem \ref{thm:carrot_stick}, implies that any partial commitment mechanism capable of achieving $\overline V$ must also admit other, less favorable equilibria.

A natural question is whether the converse of Proposition \ref{prop:virtual_revelation} holds. That is, can the set of equilibrium outcomes of \emph{any} partial commitment mechanism be exactly replicated by a full-commitment mechanism restricted to the corresponding set of equilibrium strategies?

The primary obstacle to such a general equivalence lies in the dual channel through which the agent's strategy affects the principal's payoff under partial commitment. An agent selecting a strategy $\sigma_A$ from the equilibrium set $\Sigma(Q)$ influences the principal's payoff by: (1) determining the probability distribution over the committed actions $\pi_m$ and uncommitted responses, and (2) determining the uncommitted principal's best-response strategy $\sigma_P(\sigma_A)$.

In a corresponding full-commitment mechanism, this second, strategic channel is absent. The principal's actions are fixed by the plan $\pi$, and the agent's choice of strategy $\sigma_A\in\Sigma(Q)$ can only affect the principal's payoff by altering the probability distribution over these fixed actions. This implies, for example, that the infimum payoff in the full-commitment characterization, taken over $\Sigma(Q)$, may be strictly lower than the true worst-case implementable payoff of the partial commitment mechanism $Q$.

However, this difficulty is resolved for the specific class of polarizing mechanisms. Under Assumption \ref{assp:genericity_2}, all strategies $\sigma_A\in\Sigma^\star$  induce the same uncommitted principal's actions. The indirect channel is therefore constant across all equilibrium strategies, and the equivalence holds. For this class of mechanisms, the set of outcomes is exactly the set of outcomes generated by a full-commitment mechanism restricted to the set of polarizing strategies $\Sigma^\star$.

\section{Proofs.}

To keep the notation simple, I denote the agent's committed and uncommitted payoffs by $u_A^{\pi}(m)=\mathbb E_{ \pi_m}[u_A(a)]$ and $u_A^{\sigma_P}(m)=\mathbb E_{ \sigma_P(m)}[u_A(a)]$. In addition by $u(\mu)$ I will mean the set of the agent's expected payoffs when the principal best responds to her posterior belief $\mu$:
\[u(\mu)=\{u\in[0,1]\mid \exists \sigma_P\in a(\mu)\quad s.t.\quad u=\mathbb E_{ \sigma_P}[u_A(a)]\}\]
Equivalently, one may  $u(\mu)$ as the principal's best response correspondence when she maximizes her effective expected payoff $\gamma_\theta(u)$ over $u\in [0,1]$.

 \subsection{Proof of Lemma \ref{lem:cheaptalk}}

An equilibrium is a profile $(\sigma_A, \sigma_P)$ satisfying two essential conditions:  Agent-Optimality (\ref{eqm:ind}), and Principal-Optimality (\ref{eqm:br}). The analysis proceeds by examining these conditions at the extremes of the commitment spectrum.

\textit{Case 1: No Commitment} ($\chi = 0$)

When $\chi=0$, the principal's strategy $\sigma_P$ is a best response to the posterior belief induced by each message. The agent-optimality condition (\ref{eqm:ind}) requires that for any on-path messages $m, m'$,
\[\mathbb{E}_{\sigma_P(m)}[u_A(a)] = \mathbb{E}_{\sigma_P(m')}[u_A(a)]\]
where $\sigma_P(m)$ is the distribution over actions taken in response to message $m$. This condition characterizes cheap talk equilibria. Hence, the principal cannot achieve a payoff above $V_{CT}$.

\textit{Case 2: Full Commitment} ($\chi = 1$)

When $\chi=1$, the principal's actions are dictated by the committed plan $\pi$. The agent-optimality condition (\ref{eqm:IC1}) becomes:
\[ m \in \argmax_{m' \in \mathcal{M}} \mathbb{E}_{\pi_{m'}}[u_A(a)] \]
If there is a single message $m_0$ that solves the above maximization, then by the no-perfect-pooling equilibrium condition $\mathcal E(Q)=\varnothing$, implying that the principal receives $V_B$. Consider a mechanism $\pi$ that satisfies agent-optimality under full commitment, i.e., $\mathbb{E}_{\pi_m}[u_A(a)]$ is constant for several messages $m$. Under such a plan, the agent is indifferent among these messages. Consequently, any messaging strategy $\sigma_A$ is a best response for the agent, and the (\ref{eqm:IC1}) condition imposes no restriction on the agent's behavior.

The principal-optimality condition (\ref{eqm:br}) is vacuous when $\chi=1$, as the principal never acts strategically. The out-of-mechanism best responses do not enter any player's utility or any other equilibrium constraint. Therefore, for any mechanism $\pi$ that satisfies agent-optimality, any messaging strategy $\sigma_A(m|\theta)$ is part of a valid equilibrium. Including the uninformative messaging $\sigma_A(m|\theta)\equiv \sigma_A(m)$ for all $\theta$. This concludes the proof.

\subsection{Proof of Proposition \ref{prop:full_commitment}}

Focusing on direct mechanisms, the problem of characterizing $\overline V$ becomes

\begin{align}
	\max_{u\in[0,1]}\max_{\pi_\theta\in\Delta([0,1])}  &\sum_{\theta}\rho(\theta)\E_{\pi_\theta}[\gamma(\tilde u,\theta)]\notag\\
s.t.& 	\notag\\
\E_{\pi_{ m}}[\tilde u] =u
%\E_{\pi_m}[u_A(a)]=&	\E_{\pi_{m'}}[u_A(a)] 	\notag
\end{align}

The problem is separable in each $\theta$ given $u$ and has a standard solution (\cite{kamenica2011bayesian}): $\overline \pi_\theta$ creates a concave hull of $\gamma_{\theta}(u)$ and the corresponding state-optimal value given $u$ is $\overline \gamma_{\theta}(u)$.

\subsection{Proof of Lemma \ref{lem:multiple_equilibria}}

Suppose $\sigma^T\in\Sigma(Q)$ for some $Q$. Since $P^T$ is truthful,  the principal's best response to message $m_i$ in $P^T$ must be $u_{\theta_i}$. The (\ref{eqm:ind}) implies
\begin{equation}
	\frac{\chi}{1-\chi}(u_{m_i}^\pi-u_{m_{i+1}}^\pi)=u_{\theta_{i+1}}-u_{\theta_i}
	\label{eqn:multiple}
\end{equation}
The LHS does not depend on an equilibrium. Thus, any strategy profile $\sigma(m|\theta)$ that satisfies the RHS for on-path messages will be an equilibrium.

All that's left is to show that the candidate strategy from the statement of the lemma satisfies (\ref{eqn:multiple}). The principal's posteriors are
\[\mu_{m_2}=P(\theta|m_2)=(0,1,0)\]
\[\mu_{m_3}=P(\theta|m_3)=\left(\frac 1 3,\frac 1 3,\frac 1 3\right)\]
And the probability of sending $m_3$ is $P(m_3)=\frac{3}{4}$. The principal's best response to $m_2$ is $u_{m_2}=0=u(\theta_2)$. The principal's best response to $m_3$ is $u_{\theta_3}$. Indeed: 
\[\gamma'(u_{\theta_3},\theta_1)+\gamma'(u_{\theta_3},\theta_2)+\gamma'(u_{\theta_3},\theta_3)=2u_{\theta_3}-2u_{\theta_3}=0\]
and 
\[\gamma''(u_{\theta_3},\theta_1)+\gamma''(u_{\theta_3},\theta_2)+\gamma''(u_{\theta_3},\theta_3)=\gamma''(u_{\theta_3},\theta_3)<0\]
where we used the fact that $u_{\theta_3}\in(u(\theta_2),u(\theta_1))$ maximizes $\gamma(u,\theta_3)$. 

To finalize the proof, I show that the principal has no incentives to deviate to the off-path message $m_1$.

	Consider some  $m$ and $\tilde m$ such that $m$ is on-path for $\sigma$ and $\sigma'$, while $\tilde m$ is on-path for $\sigma$ and off-path for $\sigma'$. Fix the principal's best responses $\sigma_P\in a(\sigma)$ and $\sigma_P'\in a(\sigma')$ such that $u_A^{\sigma_P}(m)=u_A^{\sigma_P'}(m)$ for all $m\in\supp (\sigma')$. The strategy $\sigma$ satisfies (\ref{eqm:ind}), which implies that
	\[\chi u_A^\pi(m)+(1-\chi)u_A^{\sigma_P}(m)=\chi u_A^\pi(\tilde m)+(1-\chi)u_A^{\sigma_P}(\tilde m)\ge \chi u_A^\pi(\tilde m)\]

	The agent's payoff from sending $\tilde m$ under $\sigma'$ does not exceed that of $m$:
	\[\chi u_A^\pi(m)+(1-\chi)u_A^{\sigma_P'}(m)=\chi u_A^\pi(m)+(1-\chi)u_A^{\sigma_P}(m)\ge \chi u_A^\pi(\tilde m)\]
	for the off-path beliefs $u_{A}^{\sigma_P}(\tilde m)=0$.	 Hence, $\sigma'\in \Sigma(Q)$.

\subsection{Proof of  Theorem \ref{thm:carrot_stick}}

\begin{lemma}[intermediate value]\label{lem:intermediate_value} For any $\mu_1,\mu_2$ and any $u_1,u_2$ such that $u_i\in u(\mu_i)$ the following holds: for any $\tilde u\in[u_1,u_2]$ and any continuous path $\mu(t)$ in $\Delta(\Theta)$ that connect $\mu_1=\mu(0)$ and $\mu_2=\mu(1)$ there exists $\tilde t\in[0,1]$ such that $\tilde u\in u( \mu(\tilde t))$.
\end{lemma}
\begin{proof}
The functions $\gamma_{\theta}(u)$ are continuous for each $\theta\in\Theta$. By Berge's Maximum Theorem, the correspondence $u(\mu)$ is upper hemicontinuous. Since the sets $u(\mu)$ are convex for each $\mu$, the graph of the correspondence $u(\cdot)$ is connected on any connected domain. The set of posteriors on the path $\mu(t)$ between $\mu_1$ and $\mu_2$, is a connected set. Therefore, its image under $u(\cdot)$ is also a connected set in $\mathbb{R}$. Since $u_1$ and $u_2$ are in this image, the entire interval $[u_1, u_2]$ must be as well.
\end{proof}

\begin{lemma}
Suppose $Q$ worst-case implements $v>V_B$. There exists $\sigma_A\in\Sigma(Q)$ such that  $\sigma_A$ has exactly two on-path messages.  %$m_1,m_2$. Moreover,  $u_{m_1}=u_{m_1}'>u_p>u_{m_2}'=u_{m_2}$
	\label{lem:binary}
\end{lemma}

\begin{proof}

%Pick some equilibrium $(\sigma_A,\sigma_P)\in\mathcal E(Q)$ and denote by $u_m=\mathbb E_{\sigma_P}[u|m]$.  Suppose $\sigma_A$ has $M>2$ on-path messages. Denote by $\{\mu_m\}$ the posterior beliefs induced by $\sigma_A$. Further, denote by $S=\Delta(\{\mu_{m}\}_{m:\sigma_A(m)>0})$ the convex hull of the posteriors $\mu_m$. Note that $\rho\in int(S)$. The set $S$ is an $M-1$-simplex.

Pick an equilibrium $(\sigma_A,\sigma_P)\in\mathcal E(Q)$ with $M>2$ on-path messages. To simplify the exposition, assume the principal's best response to any equilibrium posterior is a unique pure action, which holds generically. Let $u_m = u_A(a(\mu_m))$ be the agent's payoff from the principal's best response to posterior $\mu_m$. Assume without loss of generality that $u_1 \le u_2 \le ... \le u_M$.

Consider beliefs $\mu_0$ and $\mu_1$ that place probability $1$ on the lowest and the highest states respectively. The agent's  uncommitted payoffs in these states are $0$ and $1$. Let $\mu(t)=(1-t)\mu_0+t\mu_1$ for $t\in[0,1]$. And define 
%\[\Pi(t,u)=\left\{\mu \in \Delta(\Theta) \mid   \exists\sigma_P \in \text{BR}(\mu) \text{ s.t. } \mathbb{E}_{\sigma_P}[u_A(a)] =u\text{ and }u \in u(\mu(t))\right\} \]
\[\Pi(t)=%\cup_{u\in u(\mu(t))}\Pi(t,u)\]
\{\mu \in \Delta(\Theta) \mid   \exists\sigma_P \in \text{BR}(\mu) \text{ s.t. } \mathbb{E}_{\sigma_P}[u_A(a)] \in u(\mu(t)) \}\] 

 The set $\Pi(t)$ is the set of the beliefs, which make the principal deliver the agent the same expected payoff as some response to $\mu(t)$. By Lemma \ref{lem:intermediate_value} for any $u\in[0,1]$ there exists some $t$ such that $u$ is the agent's uncommitted payoff when the principal's posterior is $\mu(t)$. 

Without loss of generality for the argument, assume that for any $t$ the set $a(\mu(t))$ contains at most two pure strategies. For any $t\in(0,1)$ and any $u\in u(\mu(t))$ the set $\Pi(t)$ 
%is a union of piecewise indifference hyperplanes $\Pi(t,u)$ that each separate the belief space $\Delta(\Theta)$ into two disjoint sets.\footnote{ That is, any continuous path that connects points $\mu\in\Delta_+$ and $\mu'\in\Delta_-$ must contain an element of $\Pi(t,u)$.}
contains a hyperplane $\Pi(t,u)$ that separates the set $\Delta(\Theta)$ into two disjoint sets $\Delta_+$ and $\Delta_-$. That is, any continuous path that connects points $\mu\in\Delta_+$ and $\mu'\in\Delta_-$ must intersect $\Pi(t,u)$. To show this, it is sufficient to consider 
%an arbitrary curve that connects $\mu_0$ and $\mu_1$ inside $\Delta(\Theta)$. Lemma \ref{lem:intermediate_value} implies that any such line must contain a belief $\tilde\mu$ such that $u\in u(\tilde \mu)$, which implies the existence of a separating $\Pi(t,u)$. This argument also implies that $\Pi(t,u)$ must contain a belief from $[\mu_0,\mu_1]$. 
 $\mu(t-\varepsilon)$ and $\mu(t+\varepsilon)$ that are sufficiently close to $\mu(t)$. By Assumption \ref{assp:genericity_1} we can choose $\varepsilon$ such that the best responses to $\mu(t\pm \varepsilon)$ must be singletons. Suppose without loss that $u(\mu(t+\varepsilon))>  u> u(\mu(t-\varepsilon))$. The fact that $u\in[u(\mu(t-\varepsilon)),u(\mu(t+\varepsilon))]$ is implied by the optimality of the three points (see Remark \ref{remark:mixed}). If one of the inequalities was weak, perturbing $\varepsilon$ would make the inequalities strict by Assumption \ref{assp:genericity_1}. If the region that attains $u$ is thick, the proof extends with proper adjustments.  Lemma \ref{lem:intermediate_value} implies that any continuous path that connects $\mu(t-\varepsilon)$ to $\mu(t+\varepsilon)$ must always contain a belief that attains $u$. This implies the existence of a separating space $\Pi(t,u)$. The fact that it is contained in $\Pi(t)$ is implied since we chose $\varepsilon$ arbitrarily small.

By UHC of the best response correspondence, the family of sets $\Pi(t,u)$ can be chosen to vary continuously with $t$ and $u\in u(\mu(t))$. Hence, any belief $\mu$ must belong to some $\Pi(t,u)$. Formally, $\Delta(\Theta)\subseteq \cup_{t\in[0,1],u\in u(\mu(t))}\Pi(t,u)$. This implies the following fact.

\textit{Fact.} The sets $\Pi(t,u)$ create a "sandwich" slicing of $\Delta(\Theta)$. That is,  two sets in this family can touch, but cannot separate one another into disjoint pieces. To show this, suppose, on the contrary, that $\Pi(t_1,u_1)$ and $\Pi(t_2,u_2)$ strictly intersect. This must happen in the interior of $\Delta(\Theta)$.
%Incomplete. What if the intersection happens on the boundary?
 If they intersect at some $\mu_0$, then there exists $\sigma,\sigma'\in a(\mu_0)$ such that $u_1=\mathbb E_\sigma[u_A(a)] $, $u_2=\mathbb E_{\sigma'}[u_A(a)] $ and $\sigma\sim_{u_P}\sigma'$. Consider the set of beliefs, at which the Principal is indifferent between $\sigma$ and $\sigma'$. It is a hyperplane (by Assumption \ref{assp:genericity_1}), that divides $\Delta(\Theta)$ into two parts: where $\sigma\succ \sigma'$ and where $\sigma'\succ \sigma$. This hyperplaneplane must strictly separate either $\Pi(t_1,u_1)$ or $\Pi(t_2,u_2)$, which contradicts the optimality of $\sigma$ or $\sigma'$ on the corresponding surface.

%Suppose $t_2>t_1$. The set $\Pi(t_1,u_1)$ slices $\Delta(\Theta)$ into $\Delta_-$ and $\Delta_+$ - disjoint. Suppose $\Delta_-$ is the set that has already been "swept" by $\Pi(t,u)$ for $t\le t_1$. Then it must be that $\Pi_-=\Pi(t_2,u_2)\cap \Delta_-\neq \varnothing$. Consider the convex hull of $\Pi_-$ and $\Pi(t_1,u_1)$. It must contain a set of positive measure that has several best responses. This is because the set "between" $\Pi_-$ and $\Pi(t_1,u_1)$ has been swept by $t\le t_1$ and has been swept by $t\in [t_1,t_2]$ by continuity of $\Pi(t,u)$. This violates Assumption \ref{assp:genericity_1}.

There must exist $t_1,...,t_M$ such that $\mu_m\in \Pi(t_m,u_m)$ for all $m$. %The actions $u_m$ are unique optimal responses to $\mu_m$, therefore $u_m\in u(\mu(t_m))$. Hence, $\Pi(t_m,u_m)\subseteq \Pi(t_m)$ for all $m$. Therefore, for any  $\mu \in \Pi(t_m,u_m)$ we have $u_m\in u(\mu)$. Finally, 
The family $\{\Pi(t_m,u_m)\}_{m\le M}$ partitions the space $\Delta(\Theta)$ into parts, that each are "sandwiched" between some two elements of $\{\Pi(t_m,u_m)\}_{m\le M}$. 

The convex hull of beliefs $\{\mu_m\}_{m=1}^M$ is a convex polytope that contains $\rho$ in its interior. By construction, the surfaces $\Pi(t_m,u_m)$ slice this polytope through its vertices. Hence, the "sandwich" partitioning extends to the convex hull of $\{\mu_m\}_{m=1}^M$. %The surfaces $\Pi(t_m,u_m)$ and $\Pi(t_{m'},u_{m'})$ cannot intersect in a way that one strictly separates the other into disjoint parts. This would contradict them containing optimal strategies. (at any point of intersection, the principal must be indifferent between the strategies from both surfaces; the indifference plane must also cut one of the surfaces into two parts, that both have to contain the same optimal strategy, which would be a contradiction, as beliefs from the opposite sides of the indifference plane must induce different optimal strategies). This implies that the surfaces $\Pi(t_m,u_m)$ partition the convex hull of $\{\mu_m\}_{m=1}^M$ into pieces, that each are sandwiched between some two $\Pi(t_m,u_m)$ and $\Pi(t_{m'},u_{m'})$. Each such piece is contained in the convex hull of the corresponding pair $\Pi(t_m,u_m)$ and $\Pi(t_{m'},u_{m'})$. 
This implies that $\rho$ belongs to the convex hull of some $\Pi(t_{m_1})$ and $\Pi(t_{m_2})$, which proves the claim.\footnote{There is another way to prove the claim. From the "sandwich" argument each point $\mu_m$ is contained in some surface, that separates the belief space and necessarily crosses the $\mu_0$-$\mu_1$ edge that connects the  degenerate beliefs corresponding to extreme actions $0$ and $1$. Each such surface only contains beliefs with a principal's best response yielding $u_m$ to the agent. Thus, we can continuously move $u_m$ towards the $\mu_0-\mu_1$ line without affecting the agent's incentives. With this the convex hull of the beliefs moves continuously and, since it contains the prior, we can move the beliefs in a way that the prior touches one of the edges of the said hull, proving the claim.}

\begin{remark}[Mixed Strategies]

\label{remark:mixed}
Below is a more detailed version of the "sandwich" argument. Consider a belief $\mu_0$ where the principal is indifferent between two actions, $a_0$ and $a_1$, and suppose her best response set is $\{a_0, a_1\}$. Let $l$ be a line of beliefs passing through $\mu_0$, parametrized by $\mu(c) = \mu_0 + c\nu$. The principal's payoff difference between the two actions along this line is:
\[ D(c) = U_P(a_0, \mu(c)) - U_P(a_1, \mu(c)) \]
Since $U_P(a_0, \mu_0) = U_P(a_1, \mu_0)$, this simplifies to:
\[ D(c) = c \sum_{\theta\in\Theta}\nu(\theta)(u_P(a_0,\theta)-u_P(a_1,\theta)) = c \cdot K \]
If the constant $K \neq 0$, then for $c$ close to zero, $a_0$ is strictly optimal on one side of $\mu_0$ and $a_1$ is strictly optimal on the other. If the constant $K=0$ for any line $l$, then the agent must be indifferent between $a_0$ and $a_1$ at any belief, which contradicts Assumption \ref{assp:genericity_1}. %Now, suppose an equilibrium has a posterior $\mu_m$ that induces a mixed strategy over $\{a_0, a_1\}$. If this posterior $\mu_m$ is perturbed along the line $l$, the best response must be a pure action. A contradiction arises if one attempts to construct a sequence of posteriors converging to $\mu_m$ that must induce a continuous sequence of pure-action best responses, while the limit point $\mu_m$ itself requires a mixed strategy. The best-response correspondence cannot be continuous in this way at a point of indifference unless the agent payoffs from the pure actions are identical, which would resolve the issue. %[[This part of the argument requires a more rigorous application of upper hemicontinuity and is not fully specified.]] 
A similar argument applies if the point of indifference $\mu_0$ lies on the boundary of the belief simplex.
\end{remark}

\end{proof}

\begin{lemma}
	
Let a mechanism $Q=(\{\pi_1,...,\pi_M\},\chi)$ be such that $u_A^{\pi}(m_1)>...>u_A^{\pi}(m_m)$. If $\mathcal E(Q)\neq \varnothing$, then for any polarizing strategy $\sigma^\star\in\Sigma^\star$ there exists a two-message equilibrium $(\sigma_A,\sigma_P)\in\mathcal E(Q)$ such that it has only messages $m_1$ and $m_2$ on-path. Moreover, $\sigma_A$ is such that $\sigma^\star$ Blackwell dominates $\sigma_A$.
\label{lem:binary_blackwell}
\end{lemma}

\begin{proof}
With pessimistic off-path beliefs ($u_A^{\sigma_P}(\tilde m)=0$ for $\tilde m \notin \{m_1, m_2\}$), an equilibrium $(\sigma_A,\sigma_P)\in\mathcal E(Q)$ with on-path messages $m_1, m_2$ must satisfy the agent-optimality condition:
\begin{equation}
u_A^{\sigma_P}(m_2)-u_A^{\sigma_P}(m_1)=\frac{\chi}{1-\chi}(u_A^{\pi}(m_1)-u_A^{\pi}(m_2)) \label{eq:two_message_1}
\end{equation}
The off-path constraints $U_A(m_1) \ge U_A(m_k)$ for $k \ge 3$ are satisfied automatically, since $u_A^\pi(m_1)>u_A^\pi(m_k)$ and the payoff from the uncommitted principal is non-negative.

Recall that by definition of polarizing equilibria, the agent's uncommitted payoff spread in any two-message equilibrium cannot exceed $\overline \Delta$. 

Case 1: $\frac{\chi}{1-\chi}(u_A^{\pi}(m_1)-u_A^{\pi}(m_2))\le \overline\Delta$.
Let $(\mu_1^\star, \mu_2^\star)$ be the posteriors induced by $\sigma^\star$. Since the required payoff difference is less than or equal to the maximal possible difference $\overline\Delta$, by Lemma \ref{lem:intermediate_value} there exist posteriors $\mu_1 \in \Delta(\mu_1^\star, \rho)$ and $\mu_2 \in \Delta(\rho, \mu_2^\star)$ and corresponding principal's best responses yielding payoffs $u_1, u_2$ that satisfy (\ref{eq:two_message_1}). Since $\mu_1, \mu_2$ can be generated as a mean-preserving contraction of $\mu_1^\star, \mu_2^\star$, there exists a corresponding strategy $\sigma_A$ that is Blackwell-dominated by $\sigma^\star$. By construction, this $\sigma_A$ supports an equilibrium of $Q$.

Case 2: $\frac{\chi}{1-\chi}(u_A^{\pi}(m_1)-u_A^{\pi}(m_2))> \overline\Delta$.
This condition implies no two-message equilibrium can have $m_1$ on path. If $\mathcal E(Q)\neq \varnothing$, then by Lemma \ref{lem:binary} there must be a two-message equilibrium with on-path messages $m_j, m_k$ for $j, k \ge 2$. Suppose these are $m_2, m_3$. The agent's off-path constraint for deviating to $m_1$ requires $U_A(m_2) \ge U_A(m_1)$, which implies:
\[u_A^{\sigma_P}(m_2)\ge \frac{\chi}{1-\chi}(u_A^{\pi}(m_1)-u_A^{\pi}(m_2)) > \overline\Delta\]
This implies that the agent's payoff from any on-path message must exceed the maximal possible spread. Since $\overline\Delta\ge u_{\rho}$ (proven below), this contradicts Assumption \ref{assp:no_cheap_talk} and completes the proof.

\textit{Proof that }$\overline\Delta\ge u_{\rho}$. Suppose without loss that $\theta_0$ is such that $u(\theta_0)=0$ and let $\mu_0$ be the degenerate belief vector that assigns probability $1$ to state $\theta_0$. By upper hemicontinuity of the best response correspondence and since $\rho\in int(\Delta(\Theta))$, there must exist a messaging strategy $\sigma(\varepsilon)$ that induces posterior beliefs $\mu_1(\varepsilon)$ and $\mu_2(\varepsilon)$ such that as $\varepsilon \rightarrow 0$ we have $\mu_1(\varepsilon)\rightarrow \mu_0$, $\mu_2(\varepsilon)\rightarrow \rho$ and $u_i(\varepsilon)\in u(\mu_i(\varepsilon))$ for $i\in\{1,2\}$ such that $u_1(\varepsilon)\rightarrow 0$ and $u_2(\varepsilon)\rightarrow u_\rho$. Then 
$$\Delta(\sigma(\varepsilon))=|u_1(\varepsilon)-u_2(\varepsilon)|\rightarrow u_\rho-0>\overline \Delta$$

Thus, by picking $\varepsilon$ small enough we find a messaging strategy with a best response spread that exceeds $\Delta$, which is a contradiction. It follows that $\mathcal E(Q) = \varnothing$.

%\cmt{Here we could use the limiting argument. If there is a three-message equilibrium, we can show that there is a two-message equilibrium by contracting the belief simplex. But even if the two-message equilibrium somehow fails IC, the equilibria of the limiting sequence still converge to the two-message equilibrium. And this would still lead to a contradiction.}

%\textbf{Remark 1.} Similarly, we can show that $\overline{\Delta}\ge 1-u_{\rho}$. The bounds can be strengthened even further, if needed. 

%\textit{Remark 2.} Assumption \ref{assp:no_cheap_talk} is excessive here. The proof applies if $u_\rho\ge \min\{u_A^{\sigma_P}(m_2),u_A^{\sigma_P}(m_3)\}$. The primary role of Assumption \ref{assp:no_cheap_talk} is to remove scenarios when $u_\rho< \min\{u_A^{\sigma_P}(m_2),u_A^{\sigma_P}(m_3)\}$ and, compared to no communication, the agent can strictly benefit from both on-path messages. Note that in this case Lemma \ref{lem:intermediate_value} implies that there exists $\tilde\mu\in\Delta(\rho,\mu_{m_3})$ such that $u_{A}^{\sigma_P}(m_2)\in u(\tilde \mu)$, which implies non-trivial cheap talk with transparent motives (otherwise $Q$ worst-case implements $V_B$).

\end{proof}

\textit{Polarizing mechanisms are sufficient.} By Assumption \ref{assp:genericity_2} any pair $\sigma,\sigma'\in\Sigma^\star$ yields the same best responses from the principal. Therefore, in any equilibrium of any polarizing mechanism the uncommitted principal's actions are $u_1^\star,u_2^\star$ and are fixed.

\begin{lemma}
	Suppose Assumption \ref{assp:genericity_2} holds. Then $\Sigma^\star$ is convex.
\end{lemma}

\begin{proof}

Take $\sigma_1,\sigma_2\in\Sigma^\star$. By Assumption \ref{assp:genericity_2}, they induce the same set of best-response actions from the uncommitted principal, $\{a_m^\star\}_{m \in \{1,2\}}$. Consider $\sigma_\lambda = \lambda \sigma_1+(1-\lambda)\sigma_2$. The posterior $\mu_\lambda(\cdot|m)$ is a convex combination of $\mu_1(\cdot|m)$ and $\mu_2(\cdot|m)$. Since $a_m^\star$ is a best response to both $\mu_1(\cdot|m)$ and $\mu_2(\cdot|m)$, it is also a best response to their convex combination $\mu_\lambda(\cdot|m)$. Thus, $\sigma_\lambda$ induces the same best-response actions and is therefore polarizing. So, $\Sigma^\star$ is convex.
\end{proof}

The remainder of the proof constructs a correspondence and applies a fixed point theorem. I will use Proposition \ref{prop:virtual_revelation} to simplify the analysis by focusing focus on polarizing mechanisms with full commitment.

\begin{enumerate}
\item For an arbitrary mechanism $Q$, and for any $\sigma_0\in\Sigma^\star$, Lemma \ref{lem:binary_blackwell} guarantees the existence of a feasible strategy $\sigma_A(\sigma_0) \in \Sigma(Q)$ that is a garbling of $\sigma_0$. This map can be constructed to be continuous.
\item For each such $\sigma_A = \sigma_A(\sigma_0)$, construct a full-commitment mechanism $Q_A = (\pi_A, 1)$ such that $V(Q_A, \sigma_0) = V(Q, (\sigma_A, \sigma_P))$. This is done by defining an intermediate plan $\tilde{\pi}=\chi\pi+(1-\chi)\sigma_P$ that combines the committed and uncommitted actions from the equilibrium of $Q$, and then defining $\pi_A$ as the plan that replicates the outcome of $\tilde{\pi}$ under $\sigma_A$ when the agent's strategy is instead $\sigma_0$. Formally, construct $\pi_A(\cdot|m)=\sum_{j}\tilde\pi(\cdot|j)\kappa(j|m)$, where $\kappa$ is defined by $\sigma_A(j|\theta)=\sum_m\kappa(j|m)\sigma_0(m|\theta)$. Note that $\kappa$ continuously depends on $\sigma_0$ and $\sigma_A$. Note also that $Q_A$ is polarizing.
\item Define a correspondence $\underline{\sigma}(Q') = \argmin_{\sigma \in \Sigma^\star} V(Q', \sigma)$.
\end{enumerate}

%\textit{Continuity argument.} Informally, in the posterior belief space the strategies $\sigma_0$ and $\sigma_A(\sigma_0)$ correspond to two nested line segments that both include $\rho$. When we continuously vary $\sigma_0$, we continuously vary the corresponding posterior pairs, hence, line segments. We can always choose the nested subsegments to also continuously change with $\sigma_0$. 
	
%Note that the plans $\pi_A$ also satisfy (\ref{eqm:IC1wc}): 
%\[\E_{\pi_{A,m}}[u_A(a)]=\E_{\pi_{A,m}}[u_A(a)]= u_0\sum_j \kappa(j|m)=\E_{\tilde\pi_{m}}[u_A(a)]\]
%And let $Q_A'=(\pi_A',1)$. Note that  $V(Q_A',\sigma_0)=V((\tilde\pi,1),(\sigma_A,\sigma_P))$. 
%This is a linear program, optimized over a convex compact set. Hence $\underline\sigma(Q_A')$ is a compact, convex set. 

The composition $\Phi(\sigma_0) = \underline{\sigma}(Q_A(\sigma_A(\sigma_0)))$ is an upper hemicontinuous, compact- and convex-valued correspondence from the compact, convex set $\Sigma^\star$ to itself. By Kakutani's fixed point theorem, it has a fixed point $\sigma_0^\star$. By construction, there exists a polarizing mechanism whose worst-case payoff (under $\sigma_0^\star$) is equivalent to the payoff from some equilibrium of the original mechanism $Q$.

\textit{Remark.} The fact that $\Phi(\sigma_0)$ is upper hemicontinuous and has compact values follows from the Berge's maximum theorem. The fact that it's convex-valued follows from $\argmin$ being a linear program over a compact, convex set.

\textit{Existence of the optimal mechanism.} 

The existence of an optimal mechanism is proven by applying the Eilenberg-Montgomery fixed point theorem. Define two correspondences. First, $s(Q_0)=\argmin_{\sigma\in\Sigma^\star}V(Q,\sigma)$, which is upper hemicontinuous, compact- and convex-valued. Hence, its values are contractible. Second, \[ Q^\star(\sigma)=\argmax_{Q\text{ is }\text{(\ref{eqm:IC1wc})}}V(Q,\sigma)\]. Note that the set of mechanisms that satisfy \text{(\ref{eqm:IC1wc})} is compact and convex. The correspondence $Q^\star(\sigma)$ is upper hemicontinuous, compact- and convex-valued. Hence, its values are contractible. Finally, in an abuse of notation, define $Q^\star(Q_0)=\cup_{\sigma\in s(Q_0)}Q^\star(\sigma)$.  This is a correspondence from the compact, contractible set of polarizing mechanisms to itself. This correspondence can be shown to have contractible values. By the Eilenberg-Montgomery theorem, a fixed point exists, which corresponds to the worst-case optimal mechanism.

\subsection{Proof of Proposition \ref{prop:separation}}

Focus on some polarized actions $a_m^\star$. The set of posterior beliefs $\mathcal N_m$, to which $a_m^\star$ is a best response, contains an intersection of an $n-1$-piecewise hyperplane with the belief space $\Delta(\Theta)$, which is an $n-1$-simplex. To keep the argument simple, suppose that this intersection is just a hyperplane, no kinks. Hence, for each $m$ the set $\mathcal N_m$ contains a convex polygon. Moreover, projecting the set $\mathcal N_{m_1}$ onto $\mathcal N_{m_2}$ via connecting each point from $\mathcal N_{m_1}$ to a point in $\mathcal N_{m_2}$ (whenever such point exists) through $\rho$, yields a subset of $\mathcal N_{m_2}$ that is a convex polygon. This polygon contains all beliefs $\mu_{m_2}$ which correspond to Bayes-plausible polarizing strategies. The vertices of this polygon correspond exactly to either $\mu_{m_2}$ or  $\mu_{m_1}$ belonging to a $\theta,\theta'$-edge of $\Delta(\Theta)$. This proves the result.

\newpage
\bibliography{literature.bib}
\bibliographystyle{aer}

 \end{document}